\documentclass[conference]{IEEEtran}

\usepackage{amsmath,amsthm,amssymb}
\usepackage{algorithm}
\usepackage{booktabs}
\usepackage{array}
\usepackage{multirow}
\usepackage{threeparttable}
\usepackage{xcolor}
\usepackage{tcolorbox}
\usepackage{tikz}
\usetikzlibrary{shapes,arrows,positioning,fit,calc}
\usepackage{hyperref}
\usepackage{cleveref}
\usepackage{needspace}
\usepackage{etoolbox}
\usepackage{balance}    

\theoremstyle{definition}
\newtheorem{theorem}{Theorem}[section]
\newtheorem{lemma}[theorem]{Lemma}
\newtheorem{proposition}[theorem]{Proposition}
\newtheorem{corollary}[theorem]{Corollary}
\newtheorem{definition}[theorem]{Definition}
\newtheorem{observation}[theorem]{Observation}
\newtheorem{conjecture}[theorem]{Conjecture}
\newtheorem{remark}[theorem]{Remark}

\crefname{theorem}{Theorem}{Theorems}
\crefname{lemma}{Lemma}{Lemmas}
\crefname{proposition}{Proposition}{Propositions}
\crefname{corollary}{Corollary}{Corollaries}
\crefname{observation}{Observation}{Observations}
\crefname{conjecture}{Conjecture}{Conjectures}
\crefname{remark}{Remark}{Remarks}
\crefname{algorithm}{Algorithm}{Algorithms}
\crefname{table}{Table}{Tables}
\crefname{figure}{Figure}{Figures}

\newenvironment{proof-sketch}{\noindent\textit{Proof Sketch.}}{}
\newenvironment{note}{\noindent\textbf{Note:}}{}

\preto\section{\needspace{6\baselineskip}}
\preto\subsection{\needspace{4\baselineskip}}
\preto\subsubsection{\needspace{4\baselineskip}}

\begin{document}
\bstctlcite{IEEEexample:BSTcontrol}

\title{Safety-Certified CRT Sparse FFT: \texorpdfstring{$\Omega(k^2)$}{Omega(k^2)} Lower Bound and \texorpdfstring{$O(N \log N)$}{O(N log N)} Worst-Case}

\author{
  Aaron R. Flouro and Shawn P. Chadwick, PhD. \\
  SparseTech, Iowa, USA \\
  research@sparse-tech.com
}

\maketitle

\begin{abstract}\mdseries
Computing Fourier transforms of $k$-sparse signals, where only $k$ of $N$ frequencies are non-zero, is fundamental in compressed sensing, radar, and medical imaging. While the Fast Fourier Transform (FFT) evaluates all $N$ frequencies in $O(N \log N)$ time, sufficiently sparse signals should admit sub-linear complexity in $N$. Existing sparse FFT algorithms using Chinese Remainder Theorem (CRT) reconstruction rely on moduli selection choices whose worst-case implications have not been fully characterized.

This paper makes two contributions. First, we establish an $\Omega(k^2)$ adversarial lower bound on candidate growth for CRT-based sparse FFT when moduli are not pairwise coprime (specifically when $m_3 \mid m_1 m_2$), implying an $O(k^2 N)$ worst-case validation cost that can exceed dense FFT time. This vulnerability is practically relevant, since moduli must often divide $N$ to avoid spectral leakage, in which case non-pairwise-coprime configurations can be unavoidable. Pairwise coprime moduli avoid the proven attack; whether analogous constructions exist for such moduli remains an open question. Second, we present a robustness framework that wraps a 3-view CRT sparse front end with lightweight certificates (bucket occupancy, candidate count) and an adaptive dense FFT fallback. For signals passing the certificates, the sparse path achieves $O(\sqrt{N} \log N + kN)$ complexity; when certificates detect collision risk, the algorithm reverts to $O(N \log N)$ dense FFT, guaranteeing worst-case performance matching the classical bound.
\end{abstract}

\begin{IEEEkeywords}
Sparse FFT, Hybrid Algorithm, Safety Certificates, Chinese Remainder Theorem, Coprime Decimation, Adaptive Fallback, Worst-Case Analysis, Average-Case Complexity
\end{IEEEkeywords}

\textbf{Notation:} All logarithms are base-2 unless stated otherwise.

\section{Introduction and Problem Statement}
\label{sec:introduction}

Computing Fourier transforms of $k$-sparse signals, where only $k$ of $N$ frequencies are non-zero, is fundamental in compressed sensing~\cite{donoho2006compressed}, radar~\cite{heckel2013identification}, and medical imaging~\cite{mosca2008digital}. While the Fast Fourier Transform (FFT)~\cite{cooley1965algorithm} evaluates all $N$ frequencies in $O(N \log N)$ time, sufficiently sparse signals should admit sub-linear complexity in $N$. Sparse Fast Fourier Transform (sFFT) methods~\cite{hassanieh2012nearly,hassanieh2012simple,iwen2013improved,plonka2016acha,bittens2017acha} exploit this by frequency-binning (aliasing) to reduce dimensionality. Probabilistic variants such as Hassanieh, Indyk, Katabi, and Price~\cite{hassanieh2012nearly,hassanieh2012simple} use randomized subsampling, phase shifts, and permutation hashes to isolate tones with high probability. Deterministic variants, including Iwen~\cite{iwen2013improved}, Plonka and Wannenwetsch~\cite{plonka2015sampta,plonka2016acha}, and Bittens, Zhang, and Iwen~\cite{bittens2017acha}, replace randomization with structured congruences and CRT-style reconstructions using coprime modular decimations~\cite{garner1959residue}.

The signal processing literature has developed an extensive parallel program on robust CRT reconstruction under noisy remainders, particularly for non-pairwise-coprime moduli. Li, Liang, and Xia~\cite{li2009robust} establish remainder-error bounds for frequency estimation, and Wang and Xia~\cite{wang2010closedform} develop a closed-form robust CRT algorithm for moduli sharing a common factor, with remainder errors up to one-quarter of the GCD admitting bounded-error reconstruction; multidimensional extensions~\cite{xiao2020mdcrt,xiao2023robust} establish MD-CRT reconstruction ranges and robust recovery guarantees under matrix-valued moduli. These results provide constructive reconstruction in the same structural regime where our $\Omega(k^2)$ adversarial lower bound applies; the present paper complements this line by characterizing the worst-case complexity of deterministic 3-view CRT gating, identifying an algorithmic vulnerability that motivates the safety certificate framework of \Cref{sec:certificates}.

\subsection{Technical Gaps}

These prior methods rely on bin-isolation and static separation or structure assumptions. Multi-dimensional extensions of sFFT~\cite{viswanathan2015sampta,iwen2013improved,plonka2016acha} extend CRT-style ideas to higher dimensions but do not include runtime verification with worst-case complexity certificates. The sFFT literature has not yet integrated explicit runtime certificates for collision detection together with a deterministic fallback mechanism that yields a formal $O(N \log N)$ worst-case bound when assumptions fail. Existing CRT-based schemes depend on moduli selection choices whose worst-case implications have not been fully characterized. When moduli must divide $N$ to avoid spectral leakage (e.g., $N = 2^{20}$ has no three pairwise coprime divisors near $\sqrt{N}$), non-pairwise-coprime configurations are sometimes unavoidable, leaving open the worst-case behavior under such moduli. \Cref{tab:complexity-comparison} summarizes the complexity landscape and positions the present work.

\begin{table}[htbp]
\centering
\caption{Complexity comparison of sparse FFT algorithms for $k$-sparse signals of length $N$}
\label{tab:complexity-comparison}
\begin{threeparttable}
\begin{tabular}{@{}lcc@{}}
\toprule
\textbf{Algorithm} & \textbf{Average-Case} & \textbf{Worst-Case} \\
\midrule
Dense FFT baseline~\cite{cooley1965algorithm} & $O(N \log N)$ & $O(N \log N)$ \\
Hassanieh et al.~\cite{hassanieh2012nearly} & $O(k \log N)$ & probabilistic$^{\dagger}$ \\
Iwen~\cite{iwen2013improved} & $O(k \log N)$ & no guarantee \\
Plonka-Wannenwetsch~\cite{plonka2016acha} & $O(k^2 \log^4 N)$$^{\ddagger}$ & $O(k^2 \log^4 N)$$^{\ddagger}$ \\
\textbf{This work (Hybrid)} & $O(\sqrt{N} \log N + kN)$ & $O(N \log N)$ \\
\bottomrule
\end{tabular}
\begin{tablenotes}\footnotesize
\item[$\dagger$] High-probability guarantees (failure $< 1/\text{poly}(N)$); no deterministic worst-case bound
\item[$\ddagger$] Deterministic via rank-structured subsampling; complexity depends on signal structure and rank assumptions~\cite{plonka2016acha}
\end{tablenotes}
\end{threeparttable}
\end{table}

\needspace{6\baselineskip}
\subsection{Contributions}

This paper addresses the worst-case gap through fixed coprime modular congruences and certificate-governed control:
\begin{enumerate}
\item \textbf{Adversarial Lower Bound}: an $\Omega(k^2)$ bound on candidate growth for CRT-based sparse FFT~\cite{garner1959residue,knuth1997art} when moduli are not pairwise coprime (specifically when $m_3 \mid m_1 m_2$), implying an $O(k^2 N)$ worst-case validation cost that can exceed dense FFT time. This identifies moduli selection as a critical design parameter.
\item Hybrid Algorithm: an algorithmically deterministic sparse FFT given fixed certificate thresholds, achieving $O(\sqrt{N} \log N + kN)$ average-case complexity with $O(N \log N)$ worst-case guarantee via adaptive dense FFT fallback.
\item \textbf{Robustness Framework}: deterministic collision detection certificates (bucket occupancy, candidate count) that trigger fallback when the sparse path encounters collision risk.
\end{enumerate}

\subsection{Problem Statement}

Given time-domain samples $x[n]$, $n = 0, \ldots, N-1$, and sparsity $k \ll \sqrt{N}$, recover the support and amplitudes of the $k$-sparse spectrum $X[f]$, $f \in \{0, \ldots, N-1\}$, under on-grid and noiseless (or high-SNR) assumptions (formalized in \Cref{sec:preliminaries}). Computing all $N$ coefficients is unnecessary when only $k \ll N$ are non-zero. The goal is to achieve sublinear FFT-based identification for $k$-sparse signals while ensuring that total computational cost remains bounded by $O(N \log N)$ in all cases, eliminating the instability of candidate explosion in deterministic sparse FFT methods.

The analysis assumes the following regime:
\begin{itemize}
\item Sparsity $k$ is sufficiently smaller than $\sqrt{N}$; for $k$ approaching $\sqrt{N}$, dense FFT is typically preferable under the validation model considered here.
\item The algorithm uses three decimation moduli $m_1, m_2, m_3 \approx \sqrt{N}$ with $m_1 m_2 \geq N$ for Chinese Remainder Theorem (CRT) reconstruction.
\item Goertzel validation operates in the noiseless or high-SNR regime where thresholding reliably distinguishes true frequencies.
\item Off-grid frequency analysis with spectral leakage is outside the scope of this paper.
\end{itemize}

The hybrid algorithm is deterministic given fixed certificate thresholds: CRT reconstruction, safety certificates, Goertzel validation, and dense FFT fallback involve no randomization. The correctness and worst-case complexity bounds do not rely on uniformity assumptions about residue distributions. Safety certificates detect collision risk, enabling adaptive fallback to $O(N \log N)$ worst-case complexity using standard FFT algorithms (e.g., Cooley--Tukey). The average-case analysis (\Cref{thm:gating-expected-case}) assumes approximate uniformity of residue distributions for non-adversarial signals to derive expected candidate counts.

\subsection{Illustrative Example: Algorithm in Action}

To illustrate the hybrid algorithm's operation, consider a toy signal with $N=16$ time-domain samples and $k=2$ non-zero frequencies at $f=3$ and $f=11$.

\textbf{Setup:} Using pairwise coprime moduli $m_1 = 4$, $m_2 = 3$, and $m_3 = 5$, we consider three residue views of the signal. Since $m_1 m_2 = 12 < N$, CRT reconstruction yields congruence classes that produce multiple candidates in $[0, N)$, which must be resolved by additional views and final validation.
\begin{align*}
f=3:  \quad & (r_1, r_2, r_3) = (3, 0, 3) \\
f=11: \quad & (r_1, r_2, r_3) = (3, 2, 1)
\end{align*}

\textbf{Step 1 -- Residue Views:} The algorithm forms three residue views corresponding to moduli $m_1 = 4$, $m_2 = 3$, and $m_3 = 5$:
\begin{itemize}
\item View 1 ($m_1=4$): a peak appears at residue $r_1=3$ (both $f=3$ and $f=11$ map to this bin, creating a collision).
\item View 2 ($m_2=3$): peaks appear at residues $r_2=0$ (from $f=3$) and $r_2=2$ (from $f=11$).
\item View 3 ($m_3 = 5$): peaks appear at residues $r_3 = 3$ (from $f = 3$) and $r_3 = 1$ (from $f = 11$).
\end{itemize}

\textbf{Step 2 -- CRT Reconstruction:} Pairing residues from Views 1 and 2 via Garner's formula yields congruence classes and corresponding candidates in $[0, N)$:
\begin{align*}
(r_1=3, r_2=0) & \rightarrow f \equiv 3 \pmod{12}:\ \{3, 15\} \\
(r_1=3, r_2=2) & \rightarrow f \equiv 11 \pmod{12}:\ \{11\}
\end{align*}

Initial candidate set: $\{3, 11, 15\}$, including the spurious candidate $f=15$.

\textbf{Step 3 -- Gating with View 3:} The third view filters candidates by checking consistency:
\begin{align*}
f=3:  \quad & 3 \bmod 5 = 3 \quad \checkmark\ \text{(residue 3 is present in View 3, keep)} \\
f=11: \quad & 11 \bmod 5 = 1 \quad \checkmark\ \text{(residue 1 is present in View 3, keep)} \\
f=15: \quad & 15 \bmod 5 = 0 \quad \times\ \text{(not in } R_3\text{, discard)}
\end{align*}

Gated candidate set: $\{3, 11\}$ (residue 0 is not present in View 3, discard)

\textbf{Step 4 -- Goertzel Validation:} Compute the exact DFT magnitudes at the surviving candidate frequencies using the Goertzel algorithm~\cite{goertzel1958algorithm} applied to the full time-domain signal:
\begin{align*}
|X[3]|  & = 5.0 \quad (\text{threshold } \tau = 0.1) \quad \checkmark \text{ true frequency} \\
|X[11]| & = 3.0 \quad (\text{threshold } \tau = 0.1) \quad \checkmark \text{ true frequency}
\end{align*}

\textbf{Step 5 -- Safety Certificate Check.} The algorithm verifies the safety conditions. Bucket occupancy: View 1 contains a residue bin with occupancy 2 at $r_1 = 3$, which remains below the threshold $\tau_{\text{bucket}} = 3$ (certificate passes). Candidate count: after gating, the candidate count is 2, which is below the threshold $\tau_{\text{candidates}} = 3k = 6$ (certificate passes).

\textbf{Result:} The certificates pass, so the sparse path succeeds. The output is $\{(f = 3, A = 5), (f = 11, A = 3)\}$. In this toy example, the sparse path has illustrative cost $O(\sqrt{16} \log 16 + 2 \cdot 16)$, compared with $O(16 \log 16)$ for dense FFT.

\textbf{Contrast -- Adversarial Case:} If the signal were adversarially constructed as in \Cref{thm:adversarial-lower-bound}, gating could produce $\Omega(k^2)$ candidates. For $k = 2$ this yields 4 candidates, which does not exceed the threshold $3k = 6$, but for larger $k$ the quadratic growth eventually triggers fallback. This example illustrates how non-unique CRT reconstruction can generate spurious candidates when $m_1 m_2 < N$, and how gating, validation, and safety certificates interact to resolve or control that ambiguity.

\subsection{The Core Challenge: Aliasing Collisions in Subsampling}
\label{sec:technical-overview}

Decimation downsamples $x[n]$ by a stride $d$, retaining every $d$-th sample~\cite{oppenheim2010discrete}. For $k$-sparse signals, choosing $m \approx \sqrt{N}$ reduces the per-view FFT stage complexity from $O(N \log N)$ to $O(\sqrt{N} \log N)$ (see \Cref{cor:optimal-decimation}).

\textbf{The Problem:} Multiple true frequencies can map to the same residue bin (collision), preventing unique identification from a single decimated view. For example, with $N=16$ and stride $d=4$ (modulus $m=4$), frequencies $f=3$ and $f=11$ both map to $r = 3 \bmod 4 = 3$.

\textbf{The Chinese Remainder Theorem (CRT) Approach:} By using multiple coprime moduli $m_1, m_2, \ldots$, each frequency produces a residue tuple $(r_1, r_2, \ldots)$~\cite{garner1959residue,knuth1997art}, which can be combined with CRT to reconstruct candidate frequencies. For example, using $m_1=4$ and $m_2=3$:
\begin{align*}
f=3:  \quad & (r_1, r_2) = (3, 0) \\
f=11: \quad & (r_1, r_2) = (3, 2) \quad \text{(distinguished by } r_2\text{)}
\end{align*}

\subsection{Limitations of the Two-View Approach}

However, as shown next, this reconstruction can produce multiple candidates in the presence of collisions.

A direct strategy pairs residues from two decimated views (moduli $m_1, m_2$) and applies Garner's CRT formula to reconstruct candidate frequencies:
\[
f = r_1 + m_1 \cdot [(r_2 - r_1) \cdot m_1^{-1} \bmod m_2]
\]
This produces $|R_1| \times |R_2|$ candidate pairs, where $R_i$ is the set of non-zero residues in view $i$. For $k$-sparse signals with collision multiplicity $c$ (average number of true frequencies per occupied residue bin), each residue in one view may correspond to up to $c$ frequencies, so pairwise reconstruction can produce up to $O(k^2)$ candidate pairs.

\textbf{Problem 1 -- False Positives:} Most candidate pairs correspond to spurious frequencies rather than true components of the signal, reflecting the non-uniqueness of CRT reconstruction under collisions. Only $k$ correspond to true frequencies. These candidates must be validated within the full time-domain signal, costing $O(k^2 N)$ under single-bin evaluation (e.g., Goertzel~\cite{goertzel1958algorithm}). For $c > 1$, this exceeds dense FFT's $O(N \log N)$ when $k^2 > \log N$.

\textbf{Problem 2 -- No Worst-Case Bound:} Prior deterministic sparse FFT work operated under uniformity assumptions on residue distributions, leaving worst-case behavior as an open question. Can adversarial signals force high collision rates?

\subsection{Third-View Gating: Typical-Case Candidate Reduction}

Adding a third decimated view (modulus $m_3$, coprime to $m_1, m_2$) introduces an additional consistency constraint for candidate frequencies. Given a candidate reconstructed from two views, its residue modulo $m_3$ can be compared against the detected residues in the third view. A candidate is retained if its residue modulo $m_3$ matches one of the detected residues in the third view. This additional constraint reduces the number of candidate frequencies under typical residue distributions. Under a heuristic assumption that residues of false candidates are approximately uniformly distributed modulo $m_3$, each candidate survives the third-view check with probability proportional to $1/m_3$.

\needspace{4\baselineskip}
\textbf{Typical-Case Analysis:} Assume false candidates produce $r_3'$ values approximately uniformly distributed over $[0, m_3)$. Then the survival probability is $|R_3| / m_3 \approx kc / \sqrt{N}$, and the expected post-gating candidate count becomes:
\[
\mathbb{E}[\text{candidates}] = k + (k^2c^2 - k) \cdot \frac{kc}{\sqrt{N}} = k + \Theta\!\left(\frac{k^3 c^3}{\sqrt{N}}\right)
\]
For $k \ll N^{1/3}$ and $c = O(1)$, this is $O(k)$, enabling $O(\sqrt{N} \log N + kN)$ total complexity.

\paragraph{Adversarial Lower Bound on Candidate Growth} Whether this typical-case bound is also worst-case is addressed next: \Cref{thm:adversarial-lower-bound} shows that for specific moduli configurations (e.g., $m_3 \mid m_1 m_2$), there exist $k$-sparse signals that generate $\Omega(k^2)$ candidate pairs even after third-view gating.

\needspace{6\baselineskip}
\subsection{Our Contribution: Adversarial Lower Bound}

\Cref{thm:adversarial-lower-bound} shows that for specific moduli configurations (e.g., $m_3 \mid m_1 m_2$), there exist $k$-sparse signals that generate $\Omega(k^2)$ candidate pairs even after third-view gating. The construction exploits algebraic structure: Garner's CRT reconstruction becomes affine modulo $m_3$, and adversarial placement of $k$ frequencies yields $r_3' \in R_3$ across all $k^2$ residue pairs $(r_1,r_2)$.

\textbf{Implication.} Under such configurations, third-view gating does not reduce candidate counts below $\Omega(k^2)$ in the worst case, resulting in $O(k^2 N)$ validation cost under single-bin evaluation (e.g., Goertzel~\cite{goertzel1958algorithm}), which can exceed the $O(N \log N)$ complexity of dense FFT~\cite{cooley1965algorithm} for sufficiently large $k$. This behavior arises from structured collisions that preserve consistency across views.

\subsection{The Hybrid Solution: Adaptive Fallback with Safety Certificates}

The hybrid algorithm addresses this through:
\begin{enumerate}
\item Sparse path: perform 3-view reconstruction and Goertzel validation.
\item Safety certificates: track runtime indicators of adversarial structure within the gated candidate set:
\begin{itemize}
\item Validate consistency of reconstructed frequencies across views.
\item Candidate count: trigger fallback if post-gating candidate count exceeds a threshold (e.g., $3k$).
\item For non-adversarial signals satisfying the certificate conditions, validate consistency of reconstructed frequencies across views.
\end{itemize}
\item Adaptive fallback: for non-adversarial signals satisfying the certificate conditions the sparse path executes; otherwise the algorithm reverts to dense FFT ($O(N \log N)$).
\end{enumerate}

\textbf{Ensuring $O(N \log N)$ worst-case complexity:}
\begin{itemize}
\item Average-case: for non-adversarial signals satisfying the certificate conditions the sparse path succeeds with $O(\sqrt{N} \log N + kN)$ complexity, while the fallback ensures $O(N \log N)$ worst-case complexity~\cite{cooley1965algorithm}.
\item Worst-case: for adversarial signals that fail the certificates, fallback ensures $O(N \log N)$ complexity, matching dense FFT.
\end{itemize}

Within the analyzed CRT-based framework, the hybrid approach provides a practical complexity tradeoff: exploiting sparsity when certificates pass, while maintaining $O(N \log N)$ worst-case bounds via deterministic fallback.

\subsection{Paper Organization}
\label{sec:paper-organization}

Section~\ref{sec:preliminaries} covers CRT preliminaries and formal notation. Section~\ref{sec:algorithm} presents the hybrid algorithm. Section~\ref{sec:lower-bound} proves the adversarial lower bound. Section~\ref{sec:main-theorem} analyzes complexity. Section~\ref{sec:certificates} details safety certificates. Section~\ref{sec:conclusion} concludes.

\section{Preliminaries and Notation}
\label{sec:preliminaries}

\subsection{Notation Table}

The following table summarizes key symbols used throughout this paper:

\begin{table}[!t]
\centering
\caption{Notation and Symbols}
\label{tab:notation}
\small
\setlength{\tabcolsep}{4pt}
\renewcommand{\arraystretch}{1.15}
\begin{tabular}{@{}l p{0.76\columnwidth}@{}}
\toprule
\textbf{Symbol} & \textbf{Description} \\
\midrule
$N$ & DFT size (signal length) \\
$k$ & Sparsity level ($|\text{supp}(X)| \leq k$) \\
$d_i$ & Decimation factor for view $i$ (e.g., as we have three views in this paper, $i \in \{1,2,3\}$) \\
$m_i$ & Modulus for view $i$ ($m_i = N/d_i$) \\
$\rho_i$ & Residue in view $i$, $\rho_i = f \bmod m_i$ \\
$R_i$ & Set of detected residues in view $i$ \\
$P$ & Product of moduli ($P = m_1 m_2$); CRT reconstruction range, distinct from collision multiplicity $M$ \\
$X[f]$ & DFT coefficient at frequency $f$ \\
$c$ & Coverage multiplier for top-$kc$ residue selection (Algorithm~\ref{alg:hybrid}); distinct from the informal collision multiplicity $c$ used in Section~\ref{sec:introduction} \\
$\tau$ & Bucket occupancy threshold (safety certificate, default $3$); threshold symbol only (not collision multiplicity) \\
\bottomrule
\end{tabular}
\end{table}

\subsection{Signal Model}

\begin{definition}[Discrete Fourier Transform~\cite{oppenheim2010discrete}]
\label{def:dft}
\begin{equation}
X[f] = \sum_{n=0}^{N-1} x[n] \cdot e^{-j2\pi fn/N}, \quad f = 0, 1, \ldots, N-1
\end{equation}
\end{definition}

\begin{definition}[$k$-Sparse Signal]
\label{def:k-sparse}
A signal $X$ is $k$-sparse if $|\text{supp}(X)| \leq k$, where $\text{supp}(X) = \{f : X[f] \neq 0\}$.
\end{definition}

\begin{definition}[On-Grid Frequencies]
\label{def:on-grid}
\looseness=-1
Frequencies $f \in \{0, 1, \ldots, N-1\}$ are on-grid. This paper analyzes on-grid frequencies; off-grid frequency analysis with spectral leakage is left to future work.
\end{definition}

\needspace{5\baselineskip}
\subsection{Decimation and Aliasing} Let $d$ be a positive integer such that $d \mid N$. Decimating the signal $x[n]$ by factor $d$ produces the sequence.

\begin{definition}[Signal Decimation]
\label{def:decimation}
Let $d$ be a positive integer such that $d \mid N$. Decimating signal $x[n]$ by factor $d$ produces the sequence
\begin{equation}
x_d[j] = x[j \cdot d], \quad j = 0, 1, \ldots, (N/d) - 1
\end{equation}
The resulting decimated signal has length $N/d$.
\end{definition}

\begin{theorem}[Decimation Aliasing]
\label{thm:decimation-aliasing}
Let $d \mid N$ and define $m = N/d \in \mathbb{N}$. Let $X[f]$ have energy at frequency bin $f$ in the $N$-point DFT; then the $(N/d)$-point DFT of the decimated signal $X_d$ has energy at residue bin:
\begin{equation}
r = f \bmod \left(\frac{N}{d}\right)
\end{equation}
\end{theorem}

\begin{proof}
The DFT of decimated signal $x_d[j] = x[jd]$ is:
\begin{equation}
X_d[r] = \sum_{j=0}^{N/d-1} x[jd] e^{-j2\pi rj/(N/d)} = \sum_{j=0}^{N/d-1} x[jd] e^{-j2\pi rjd/N}
\end{equation}
This sums contributions from frequencies $f \equiv r \pmod{N/d}$ in the original spectrum, i.e., $f = r + q(N/d)$ for integer $q$. Thus decimation aliases all frequencies separated by $N/d$ onto the same bin $r$ in the reduced spectrum.
\end{proof}

\begin{remark}
Decimation induces a partition of the frequency domain into equivalence classes modulo $m = N/d$. All frequencies $f$ satisfying $f \equiv r \pmod{m}$ contribute to the same residue bin $r$, resulting in collisions that obscure individual frequency identities.
\end{remark}

\subsection{Chinese Remainder Theorem for Candidate Reconstruction and Validation}

\begin{theorem}[Chinese Remainder Theorem~\cite{knuth1997art,garner1959residue}]
\label{thm:crt-reconstruction}
Let $m_1, m_2$ be coprime positive integers. For any residues $r_1 \in \{0, \ldots, m_1 - 1\}$ and $r_2 \in \{0, \ldots, m_2 - 1\}$, there exists a unique solution $f \in \{0, \ldots, M - 1\}$ with $M = m_1 m_2$ satisfying:
\begin{align}
f &\equiv r_1 \pmod{m_1} \\
f &\equiv r_2 \pmod{m_2}
\end{align}
\end{theorem}

\begin{proof}
By Bezout's identity, $\gcd(m_1, m_2) = 1$ implies existence of integers $a, b$ such that $a \cdot m_1 + b \cdot m_2 = 1$. The solution is:
\begin{equation}
f = (r_1 \cdot b \cdot m_2 + r_2 \cdot a \cdot m_1) \bmod (m_1 \cdot m_2)
\end{equation}
Substituting into the congruences verifies that this expression satisfies both modular conditions.
Uniqueness follows from the fact that any two solutions differ by a multiple of $m_1 \cdot m_2$.
\end{proof}

\textbf{Application to Sparse FFT:} For $N$-point DFT with decimation factors $d_1, d_2, d_3$ where $d_i | N$ and moduli $m_i = N/d_i$ satisfy $\gcd(m_i, m_j) = 1$ for all $i \neq j$:
\begin{itemize}
\item View 1: Decimation by $d_1$ gives moduli $m_1 = N/d_1$
\item View 2: Decimation by $d_2$ gives moduli $m_2 = N/d_2$
\item View 3 (gating): Decimation by $d_3$ gives moduli $m_3 = N/d_3$
\item If $m_1 m_2 \geq N$, CRT uniquely reconstructs a candidate frequency $f \in [0, N)$ from residue pairs $(r_1, r_2)$, provided the residues correspond to a single underlying frequency.
\item View 3 provides an additional consistency constraint: candidates are retained if their reconstructed residues agree across all three moduli, filtering many false pairs. However, in the presence of structured collisions, multiple candidates may satisfy all congruence constraints.
\end{itemize}

\subsection{Coprime Selection for \texorpdfstring{$O(\sqrt{N} \log N)$}{$O(sqrt(N) log N)$}}

\begin{corollary}[Sublinear FFT Cost via Coprime Decimation]
\label{cor:optimal-decimation}
Choosing moduli $m_1, m_2, m_3 \approx \sqrt{N}$ pairwise coprime with decimation factors $d_i = N/m_i$ (thus $d_i | N$) yields:
\begin{itemize}
\item View 1 FFT size: $m_1 \approx \sqrt{N}$
\item View 2 FFT size: $m_2 \approx \sqrt{N}$
\item View 3 (gating) FFT size: $m_3 \approx \sqrt{N}$
\item Total FFT cost: 3 $\times$ $O(\sqrt{N} \log \sqrt{N})$ = $O(\sqrt{N} \log N)$
\end{itemize}
\end{corollary}

\begin{proof}
FFT of size $m$ costs $m \log m$ operations. For $m = \sqrt{N}$:
\begin{align}
\text{Cost} &= 3 \times \sqrt{N} \times \log(\sqrt{N}) \\
&= 3 \times \sqrt{N} \times \frac{1}{2} \log N \\
&= \frac{3}{2} \sqrt{N} \log N = O(\sqrt{N} \log N)
\end{align}
This is sub-linear: $\sqrt{N} \log N \ll N \log N$ for large $N$.
\end{proof}

\textbf{Practical Note on Coprime Selection:} Finding three pairwise coprime moduli $m_1, m_2, m_3 \approx \sqrt{N}$ that satisfy $d_i = N/m_i \in \mathbb{N}$ (i.e., $m_i \mid N$) can be challenging in practice. This requirement restricts the algorithm to values of $N$ with rich factorizations. For arbitrary $N$, relaxing the divisibility constraint and using moduli $m_i \approx \sqrt{N}$ (not necessarily dividing $N$) introduces additional aliasing effects beyond the ideal modular model. In this case, the exact CRT reconstruction model no longer applies directly, and these effects must be handled via additional per-candidate validation (e.g., Goertzel evaluation)~\cite{goertzel1958algorithm}. Alternatively, padding $N$ to a highly composite number enables easier coprime selection at the cost of additional FFT operations. These practical considerations are implementation-dependent and beyond the theoretical scope of this complexity analysis.

\needspace{8\baselineskip}
\section{Algorithm Overview}
\label{sec:algorithm}

\subsection{High-Level Workflow}

\begin{algorithm}[t]
\caption{Hybrid Sparse-Dense FFT with Safety Certificates}
\label{alg:hybrid}
\textbf{Input:} Signal $x[n]$ (length $N$), sparsity $k$, decimation factors $(d_1, d_2, d_3)$, coverage multiplier $c$ (default: 2).\\
\textbf{Output:} Set of $k$ spectral peaks $\{(f_i, |X[f_i]|)\}$.
\begin{enumerate}
  \item \textbf{Phase 1 (Decimated FFTs):} Compute three decimated views by selecting coprime decimation factors $d_1, d_2, d_3 \approx \sqrt{N}$. For each view $i \in \{1,2,3\}$, decimate the signal, compute FFT of length $m_i = N/d_i$, and extract top-$ck$ residues to form $R_i$.

  \item \textbf{Phase 2 (Safety Certificate -- Bucket Occupancy):} Compute the maximum bucket occupancy $\tau = \max_{i,r} |\{f : f \bmod m_i = r, f \in R_i\}|$ across all three views. If $\tau > 3$, trigger dense FFT fallback and return $\mathrm{TopK}(\mathrm{FFT}(x), k)$ with $O(N \log N)$ complexity.

  \item \textbf{Phase 3 (CRT Reconstruction with 2-of-3 Gating):} Extract residue set $R_3^{\text{res}} = \{r : (r, \text{mag}) \in R_3\}$. Initialize candidate set $C = \emptyset$. For each pair of residues $(r_1, \text{mag}_1) \in R_1$ and $(r_2, \text{mag}_2) \in R_2$, reconstruct frequency $f = \mathrm{Garner}(r_1, r_2, m_1, m_2)$ using two-modulus CRT. Compute third-view residue $r_3' = f \bmod m_3$ and check gating condition: if $r_3' \in R_3^{\text{res}}$, add $(f, \min(\text{mag}_1, \text{mag}_2))$ to $C$.

  \item \textbf{Phase 3b (Safety Certificate -- Candidate Count):} If $|C| > 3k$, detect candidate explosion indicating adversarial signal structure. Trigger dense FFT fallback and return $\mathrm{TopK}(\mathrm{FFT}(x), k)$ with $O(N \log N)$ complexity.

  \item \textbf{Phase 4 (Goertzel Validation):} For each candidate $f \in C$, compute exact DFT magnitude $|X[f]|$ using Goertzel's algorithm with $O(N)$ complexity per bin. Retain validated frequencies in set $S$.
\end{enumerate}
\textbf{Return:} $\mathrm{TopK}(S, k)$, the $k$ largest magnitude spectral peaks.
\end{algorithm}

\begin{figure}[t]
\centering
\resizebox{0.97\columnwidth}{!}{%
\begin{tikzpicture}[
  node distance=0.9cm and 0.9cm,
  startstop/.style={rectangle, rounded corners, minimum width=2.4cm, minimum height=0.75cm, text centered, draw=black, inner sep=2pt},
  process/.style={rectangle, minimum width=2.4cm, minimum height=0.75cm, align=center, draw=black, inner sep=2pt},
  decision/.style={diamond, aspect=2.2, minimum width=2.4cm, minimum height=0.75cm, align=center, draw=black, inner sep=1pt},
  arrow/.style={thick,->,>=stealth},
  label/.style={font=\footnotesize}
]

\node (start) [startstop] {Input: $x[n]$, $k$};
\node (phase1) [process, below=of start] {3-View FFTs\\$O(\sqrt{N} \log N)$};
\node (cert1) [decision, below=of phase1] {Bucket\\occupancy\\$\leq 3$?};
\node (crt) [process, below=of cert1] {CRT Gating\\$(kc)^2 \to O(k)$};
\node (cert2) [decision, below=of crt] {Candidates\\$\leq 3k$?};
\node (goertzel) [process, below=of cert2] {Goertzel\\$O(kN)$};
\node (output) [startstop, below=of goertzel] {Output: $k$ peaks};
\node (fallback) [process, right=1.4cm of cert1] {Dense FFT\\$O(N \log N)$};

\draw [arrow] (start) -- (phase1);
\draw [arrow] (phase1) -- (cert1);
\draw [arrow] (cert1) -- node[anchor=east,label] {Yes} (crt);
\draw [arrow] (crt) -- (cert2);
\draw [arrow] (cert2) -- node[anchor=east,label] {Yes} (goertzel);
\draw [arrow] (goertzel) -- (output);

\draw [arrow] (cert1.east) -- node[anchor=south,label] {No} (fallback.west);
\draw [arrow] (cert2.east) -| node[anchor=north,label,pos=0.08] {No} (fallback.south);
\draw [arrow] (fallback.east) -- ++(0.4,0) |- (output.east);

\end{tikzpicture}%
}
\caption{Hybrid algorithm flowchart with adaptive fallback. Safety certificates trigger dense FFT~\cite{cooley1965algorithm} when the candidate count exceeds threshold; the bottleneck is the candidate count, not the FFT cost, and the fallback yields $O(N \log N)$ worst-case runtime.}
\label{fig:hybrid-flowchart}
\end{figure}
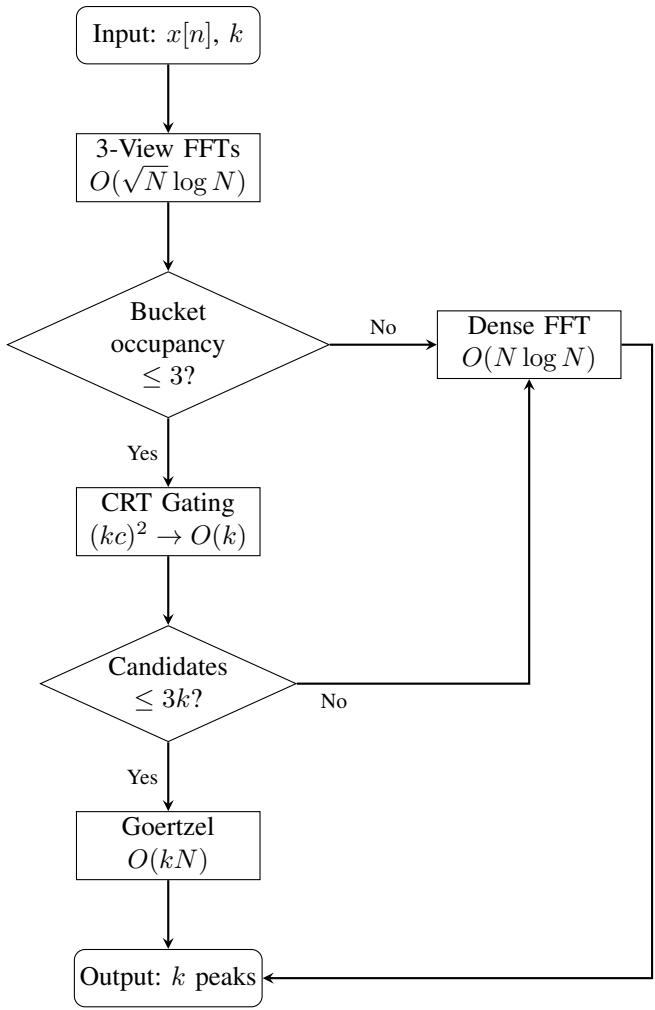

Note on complexity bottleneck. The sparse identification phase (Phases 1--3) attains $O(\sqrt{N} \log N)$ sub-linear cost; the dominant cost of the algorithm is the $O(kN)$ Goertzel validation phase (Phase 4b)~\cite{goertzel1958algorithm}. The bottleneck is therefore the candidate count, not the FFT cost, and the algorithm's advantage over dense FFT~\cite{cooley1965algorithm} depends on keeping the number of validated candidates small (specifically $k \ll \log N$). Heap-based selection attains $O(n \log k)$ worst-case time for selecting the top $k$ of $n$ elements~\cite{knuth1997art}.

\textbf{Note on Top-k Selection:} The $\mathrm{TopK}(k)$ operation selects the $k$ largest-magnitude elements from a collection of values. Standard implementations include:
\begin{itemize}
\item Quickselect algorithm~\cite{knuth1997art}: $O(n)$ average-case and $O(n^2)$ worst-case for selecting from $n$ elements.
\item Heap-based selection: $O(n \log k)$ worst-case time for selecting the top $k$ of $n$ elements~\cite{knuth1997art}.
\item Partial sort: $O(n \log k)$ using a heap or $O(k \log n)$ using a priority queue~\cite{knuth1997art}.
\end{itemize}

In the regime of interest, the cost of these selection routines is lower-order relative to candidate validation.

\subsection{Decimation Function}

For a decimation factor $d \mid N$, the decimation operator extracts every $d$-th sample of $x[n]$:
\begin{equation}
x_d = \text{decimate}(x, d) = \{x[0], x[d], x[2d], \ldots, x[d \cdot (N/d - 1)]\}
\end{equation}

Complexity: $O(N/d)$ single pass extracting $N/d$ samples with stride $d$.

\subsection{Top-k Residue Selection}

After computing decimated FFT $X_d$ of size $N/d$, the algorithm selects the top $kc$ bins by magnitude in each view. This selection is performed independently for each decimated view. Here, $c$ is a coverage multiplier controlling over-selection of residue bins to mitigate aliasing, spectral leakage, and noise effects.

\textbf{Coverage Multiplier Selection:}

The coverage multiplier $c$ must account for:
\begin{enumerate}
\item A true frequency may spread energy across multiple neighboring residue bins due to leakage.
\item Noise robustness: additional margin may be needed to account for magnitude perturbations caused by noise.
\item $c = 2$ is typically sufficient, since each true frequency maps to a single dominant residue bin under the exact-decimation model considered in this paper.
\end{enumerate}

\needspace{5\baselineskip}
\textbf{Practical Heuristic Values:}
\begin{itemize}
\item On-grid, high-SNR signals (this paper's analysis regime, \Cref{def:on-grid}): $c = 2$ is typically sufficient, since each true frequency maps to a single dominant residue bin under the exact-decimation model considered in this paper. Larger values, such as $c \in [5, 20]$, may be required to capture broadened residue structure. However, increasing $c$ also enlarges the expected false-positive term $k^3 c^3 / N$, and may trigger dense fallback in many cases.
\item Off-grid or noisy signals (beyond this paper's scope): larger values, such as $c \in [5, 20]$, may be required to capture broadened residue structure. Increasing $c$ also enlarges the expected false-positive term $k^3 c^3 / \sqrt{N}$ (see \Cref{thm:gating-expected-case}) and may trigger dense fallback in many cases.
\end{itemize}

\textbf{Trade-off Analysis:} Larger coverage multipliers $c$ increase capture probability but incur $O(c^3)$ cost in the expected false positive term $k^3 c^3 / \sqrt{N}$ (\Cref{thm:gating-expected-case}). The complexity bounds in this paper assume $c = O(1)$, which is justified for the on-grid noiseless regime analyzed here. Extending the analysis to large $c$ regimes is left to future work.

\begin{observation}[Coverage Multiplier Heuristic]
\label{obs:coverage}
Under bounded spectral leakage assumptions, selecting the top $kc$ residues with coverage multiplier $c \geq \lceil \text{max aliasing factor} \rceil$ is sufficient with probability $\geq 1 - \delta$ (for small $\delta$) to capture all $k$ true frequencies.
\end{observation}

\begin{note}
Each true frequency $f$ aliases to at most $c$ residues across the $d_i$ decimation factors due to spectral leakage. Selecting top $kc$ residues ensures all $k$ true frequencies and their aliases are captured, with probability approaching 1 as signal-to-noise ratio increases. This is a heuristic bound based on empirical observations; a rigorous probabilistic model requires formal stochastic signal assumptions beyond the scope of this work.
\end{note}

\begin{theorem}[Goertzel Disambiguation]
\label{thm:goertzel-disambig}
For an on-grid frequency $f$, Goertzel's algorithm computes the exact DFT magnitude $|X[f]|$. When $M < N$, CRT reconstruction produces a set of alias candidates $f' = f_0 + qM$ within $[0, N)$. Evaluating $|X[f']|$ for each candidate yields the true DFT magnitude at that frequency. In the noiseless sparse setting, only candidates corresponding to true frequencies exhibit magnitudes consistent with the signal support, allowing incorrect aliases to be discarded.
\end{theorem}
\begin{proof}
For on-grid frequency $f$, Goertzel computes $|X[f]|$ exactly. Aliases $f' = f + qM$ have different DFT values due to phase differences across the full $N$-point signal. In noiseless case, only the true $f$ produces magnitude matching the decimated view estimates.
\end{proof}

\subsection{Garner's CRT Formula with 3-View Gating}

\begin{lemma}[Garner's 2-Residue CRT with Gating Validation]
\label{lem:garner-3view}
For pairwise coprime $m_1, m_2, m_3$ and residues $(r_1, r_2)$, the candidate frequency $f \in [0, m_1 \times m_2)$ is:
\begin{equation}
f = r_1 + m_1 \times \left[ (r_2 - r_1) \times m_1^{-1} \bmod m_2 \right]
\end{equation}
This candidate is validated by checking if $f \bmod m_3 \in R_3$ (2-of-3 CRT agreement).
\end{lemma}

\begin{proof}
\textbf{Step 1: CRT Reconstruction (Views 1 \& 2):}

\textbf{Verification of $f \equiv r_1 \pmod{m_1}$:}
\begin{equation}
f = r_1 + m_1 \times u \implies f \bmod m_1 = r_1 \bmod m_1 = r_1
\end{equation}

\textbf{Verification of $f \equiv r_2 \pmod{m_2}$:}
Let $\gamma = m_1^{-1} \bmod m_2$. Then $m_1 \times \gamma \equiv 1 \pmod{m_2}$.
\begin{align}
u &= (r_2 - r_1) \times \gamma \bmod m_2 \\
m_1 \times u &\equiv (r_2 - r_1) \times \gamma \times m_1 \pmod{m_2} \\
&\equiv (r_2 - r_1) \pmod{m_2}
\end{align}
where the last step follows from $\gamma \times m_1 \equiv 1 \pmod{m_2}$.
Therefore:
\begin{equation}
f = r_1 + m_1 \times u \equiv r_1 + (r_2 - r_1) \equiv r_2 \pmod{m_2}
\end{equation}

\textbf{Step 2: Gating Validation (View 3):}

Compute $\rho_3' = f \bmod m_3$. The candidate $f$ passes gating if and only if $\rho_3' \in R_3^{\text{res}}$, where $R_3^{\text{res}}$ is the set of residues extracted from detected peaks in view 3. This enforces 2-of-3 CRT agreement: views 1 and 2 reconstruct $f$ and view 3 checks consistency. Complexity: the modular inverse $\gamma$ can be computed once in $O(\log m_2)$ time, and each candidate reconstruction and consistency check requires $O(1)$ time~\cite{garner1959residue}.
\end{proof}

Complexity: The modular inverse $\gamma$ can be computed once in $O(\log m_2)$ time. Each candidate reconstruction and consistency check requires $O(1)$ time.

\subsection{Goertzel Algorithm}

\begin{lemma}[Goertzel Single-Bin DFT]
\label{lem:goertzel}
The Goertzel algorithm computes the DFT magnitude at a single frequency bin $f$ in $O(N)$ operations without computing the full FFT.
\end{lemma}

\begin{proof}
The Goertzel algorithm uses a second-order recurrence to efficiently compute a single DFT bin without complex exponentials at every step.

\textbf{Standard Real-Arithmetic Formulation:}

For frequency bin $k$, define:
\begin{align}
\omega &= 2\pi k / N \\
\text{coeff} &= 2 \cos(\omega)
\end{align}

Initialize $s_0 = 0$, $s_1 = 0$. Then for $n = 0, \ldots, N-1$:
\begin{align}
s_2 &= s_1 \\
s_1 &= s_0 \\
s_0 &= x[n] + \text{coeff} \times s_1 - s_2
\end{align}

After $N$ iterations, the DFT magnitude is:
\begin{equation}
|X[k]| = \sqrt{s_0^2 + s_1^2 - 2 s_0 s_1 \cos(\omega)}
\end{equation}

Complexity. Each iteration requires constant work, giving $O(N)$ operations per frequency bin.

\textbf{Correctness:} The recurrence implements the DFT definition via polynomial evaluation at $z = e^{j\omega}$ on the unit circle. The final combination eliminates the filter's pole, yielding exact DFT value $X[k]$ subject only to floating-point precision~\cite{goertzel1958algorithm,oppenheim2010discrete}.
\end{proof}

Critical Observation: Goertzel is efficient for computing a small number of bins. However, for $m$ bins, total cost is $O(mN)$. When $m = k^2$, cost becomes $O(k^2 N)$, which is the bottleneck in this algorithm, since candidate validation dominates the cost.

\needspace{6\baselineskip}
\subsection{Safety Certificate Mechanisms}

\begin{definition}[Safety Certificates]
\label{def:certificates}
A safety certificate is a deterministic check designed to detect structured collision patterns in the sparse reconstruction process. Certificates may produce false positives (triggering fallback), and are calibrated to detect known adversarial constructions rather than provide worst-case optimal guarantees.
\end{definition}
\noindent Indicators of collision density: the per-view count $\tau_i(r)$ denotes the number of candidate frequencies mapping to residue bin $r$ in view $i$, and the bucket-occupancy certificate is defined as $\max_{i,r} \tau_i(r)$. The threshold value $3$ is empirically chosen to flag high collision density without incurring excessive false positives on typical signals.

\begin{theorem}[Certificate Design]
\label{thm:certificate-design}
The following certificates provide indicators of adversarial collision patterns:

\textbf{Certificate 1: Bucket Occupancy (Pre-Gating)}
\begin{align}
\text{maxBucketOccupancy} &= \max_{i \in \{1,2,3\}} \max_{r} \Big|\{r' \in R_i : r' = r\}\Big|
\end{align}
Let $\mathcal{C}$ denote the set of candidates after gating. If $|\mathcal{C}| > 3k$, candidate explosion is detected.
where $R_i$ is the set of top-$kc$ residues detected in decimated view $i$ (from Phase 1).
\textbf{Threshold:} If $\text{maxBucketOccupancy} > 3$, high aliasing detected.

Certificate 2: Post-Gating Candidate Count. Let $\mathcal{C}$ denote the set of candidates after gating. If $|\mathcal{C}| > 3k$, candidate explosion is detected.
\begin{equation}
|\text{candidates after gating}|
\end{equation}
\textbf{Threshold:} If $> 3k$, candidate explosion detected.
\end{theorem}

\begin{proof}
Certificate Calibration: the threshold constants are calibrated against sparsity $k$.

For typical $k$-sparse signals with well-separated frequencies:
\begin{itemize}
\item Bucket occupancy remains small (typically $\leq 2$)
\item Candidate count after gating $= \Theta(k + k^3c^3/\sqrt{N}) \ll 3k$
\end{itemize}

For adversarial signals (\Cref{thm:adversarial-lower-bound}):
\begin{itemize}
\item Arithmetic progression construction uses distinct residues in each view (bucket occupancy = 1), bypassing the bucket occupancy certificate
\item Post-gating candidates $\approx 0.75k^2$ (for k=12: 108 candidates = 9k), triggering the candidate count certificate
\end{itemize}

\textbf{Conservative Threshold Selection:}

Setting thresholds $\tau = 3$ (bucket) and $3k$ (post-gating) is intended to provide:
\begin{itemize}
\item High-probability acceptance for typical signals (sparse path succeeds)
\item Detection of simple collision patterns via bucket occupancy
\item Detection of adversarial candidate growth via candidate count
\item Graceful degradation to $O(N \log N)$ via fallback in false-positive cases
\end{itemize}

\textbf{Note:} These are conservative heuristics calibrated to known adversarial patterns, not proven optimal bounds. Tighter thresholds may exist but require case-by-case analysis of signal structure.
\end{proof}

\begin{corollary}[Adaptive Correctness]
The hybrid algorithm returns correct results by combining sparse reconstruction with deterministic fallback:
\begin{itemize}
\item If the sparse reconstruction path produces a consistent candidate set, exact evaluation (e.g., via Goertzel) recovers the true frequencies.
\item If the sparse path fails (as detected by the safety certificates), the algorithm falls back to a dense FFT, ensuring $O(N\log N)$ worst-case correctness.
\end{itemize}
Thus, correctness is guaranteed regardless of signal structure.
\end{corollary}

\section{Core Lemmas and Proofs}
\label{sec:lower-bound}

\subsection{Adversarial Lower Bound for 3-View Gating}

Prior deterministic sparse FFT work derived expected-case bounds of $O(k + k^3/\sqrt{N})$ candidates for 3-view CRT gating under uniformity assumptions on residue distributions. An open question is whether worst-case behavior differs significantly from this expectation. The following theorem constructs an explicit family of $k$-sparse signals that defeat gating by exploiting algebraic structure in the CRT reconstruction map, proving that $\Omega(k^2)$ candidate explosion is unavoidable for certain moduli configurations.

\begin{theorem}[Adversarial Lower Bound -- Informal]
\label{thm:adversarial-informal}
For certain choices of decimation moduli, there exist $k$-sparse signals that force the sparse FFT algorithm to examine $\Omega(k^2)$ candidates instead of the expected $O(k)$ candidates. When $k^2 > N \log N$, validation can cost $O(k^2 N)$ operations, which exceeds dense FFT's $O(N \log N)$ performance. This worst case cannot be avoided by any deterministic gating strategy using these moduli, establishing a structural limit of pure sparse approaches.
\end{theorem}

\begin{theorem}[Adversarial Lower Bound -- Formal]
\label{thm:adversarial-lower-bound}
For deterministic 3-view CRT gating with moduli $m_1, m_2, m_3 \approx \sqrt{N}$ satisfying $\gcd(m_1, m_2) = 1$ and $m_3 \mid (m_1 m_2)$, there exist $k$-sparse signals that produce $\Omega(k^2)$ candidates passing the gating check, resulting in $O(k^2 N)$ Goertzel validation cost. Such moduli can be constructed (e.g., $m_1 = g_1 m_1'$, $m_2 = g_2 m_2'$, $m_3 = g_1 g_2$ with $g_1, g_2, m_1', m_2'$ distinct primes), though $m_1, m_2, m_3$ are not pairwise coprime. This demonstrates that certain choices of moduli admit adversarial signal constructions, establishing a worst-case $\Omega(k^2)$ lower bound for CRT-based sparse reconstruction with deterministic gating.
\end{theorem}

\begin{proof}
\textbf{Adversarial Construction:}

Choose moduli $m_1, m_2, m_3 \approx \sqrt{N}$ satisfying $\gcd(m_1, m_2) = 1$ and the divisibility constraint:
\begin{equation}
m_3 \text{ divides } (m_1 \cdot m_2)
\end{equation}

\textbf{Construction Method:} Select distinct primes $g_1, g_2, m_1', m_2'$ (ensuring pairwise coprimality), then define:
\begin{align}
m_1 &= g_1 \cdot m_1' \\
m_2 &= g_2 \cdot m_2' \\
m_3 &= g_1 \cdot g_2
\end{align}
This ensures $m_3 | (m_1 m_2)$ because $m_1 m_2 = (g_1 m_1')(g_2 m_2') = m_3(m_1' m_2')$, while $\gcd(m_1, m_2) = 1$ since $g_1, g_2, m_1', m_2'$ are pairwise coprime. Note that $\gcd(m_1, m_3) = g_1 > 1$ and $\gcd(m_2, m_3) = g_2 > 1$, so the moduli are not pairwise coprime.

\textbf{Example:} For $N=10^6$, choose $g_1=23, g_2=29, m_1'=47, m_2'=37$ (all distinct primes):
\begin{align}
m_1 &= 23 \times 47 = 1081 \\
m_2 &= 29 \times 37 = 1073 \\
m_3 &= 23 \times 29 = 667
\end{align}
Verification: $\gcd(1081, 1073) = 1$; $m_1 m_2 = 1{,}159{,}913 = 667 \times 1739$, confirming $m_3 | (m_1 m_2)$.

The condition $m_3 \mid (m_1 m_2)$ with $\gcd(m_1,m_2)=1$ is precisely what allows the two-view CRT reconstruction map to reduce to an affine map modulo $m_3$, which is the mechanism exploited by the adversarial construction. The following lemma formalizes this structure.

\begin{lemma}[Affine Structure of CRT Reconstruction Modulo $m_3$]
\label{lem:crt-linearity}
Under the constraint $m_3 \mid (m_1 m_2)$ with $\gcd(m_1, m_2) = 1$, Garner's CRT reconstruction from the first two views is an affine function modulo $m_3$:
\begin{equation}
f \equiv \phi(r_1, r_2) \pmod{m_3}
\end{equation}
where $\phi(r_1, r_2) = u \cdot r_1 + v \cdot r_2 \pmod{m_3}$ for constants:
\begin{align}
\gamma &= m_1^{-1} \bmod m_2 \\
v &= m_1 \gamma \bmod m_3 \\
u &= (1 - v) \bmod m_3
\end{align}
\end{lemma}

\begin{proof}
Garner's formula gives the unique $f \in [0, m_1 m_2)$ reconstructed from residues $(r_1, r_2)$ satisfying $f \equiv r_1 \pmod{m_1}$ and $f \equiv r_2 \pmod{m_2}$:
\begin{equation}
f = r_1 + m_1 \cdot [(r_2 - r_1) \cdot \gamma \bmod m_2]
\end{equation}
where $\gamma = m_1^{-1} \bmod m_2$ (exists since $\gcd(m_1, m_2) = 1$).

Expanding modulo $m_3$:
\begin{align}
f &\equiv r_1 + m_1 \cdot [(r_2 - r_1) \cdot \gamma \bmod m_2] \pmod{m_3}
\end{align}

The key step uses $m_3 \mid m_1 m_2$. Since $m_1 m_2 \equiv 0 \pmod{m_3}$, the $\bmod\; m_2$ inside the brackets can be dropped when reducing modulo $m_3$:

For any integer $q$, $(r_2 - r_1) \cdot \gamma \bmod m_2 = (r_2 - r_1) \cdot \gamma - q \cdot m_2$ for some $q \in \mathbb{Z}$. Then:
\begin{align}
m_1 \cdot [(r_2 - r_1)\gamma - q m_2] &= m_1 \gamma (r_2 - r_1) - q m_1 m_2
\end{align}
Since $m_3 \mid m_1 m_2$, we have $q m_1 m_2 \equiv 0 \pmod{m_3}$. Therefore:
\begin{align}
f &\equiv r_1 + m_1 \gamma (r_2 - r_1) \pmod{m_3} \notag \\
&\equiv r_1 (1 - m_1 \gamma) + r_2 (m_1 \gamma) \pmod{m_3} \notag \\
&\equiv u \cdot r_1 + v \cdot r_2 \pmod{m_3}
\end{align}
where $v = m_1 \gamma \bmod m_3$ and $u = (1 - v) \bmod m_3$. Both are constants depending only on $m_1, m_2, m_3$. Such step sizes exist because this congruence defines a relation in $\mathbb{Z}_{m_3}$ with fixed coefficients $u, v$.
\end{proof}

\textbf{Step 1: Construct Aligned Residue Sets}

Define arithmetic progressions:
\begin{align}
R_1 &= \{a_0 + i \cdot d_A \bmod m_1 : i = 0, 1, \ldots, k-1\} \\
R_2 &= \{b_0 - j \cdot d_B \bmod m_2 : j = 0, 1, \ldots, k-1\}
\end{align}

Choose step sizes $d_A, d_B$ (not chosen independently, but linked across $i$ and $j$) such that:
\begin{equation}
u \cdot d_A \equiv v \cdot d_B \pmod{m_3}.
\end{equation}
Such step sizes exist because this congruence defines a linear relation in $\mathbb{Z}_{m_3}$ with fixed coefficients $u, v$.

\textbf{Step 2: Observed Gating Behavior}

For any pair $(a_i, b_j)$, if the pair passes gating:
\begin{equation}
\phi(a_i, b_j) = \text{base} + (i - j) \cdot L \pmod{m_3}
\end{equation}
where $\text{base} = u \cdot a_0 + v \cdot b_0$ and $L = u \cdot d_A \equiv v \cdot d_B$.

Thus, $\phi$ depends ONLY on the difference $(i-j)$, not individual indices.

\textbf{Step 3: Construct k True Frequencies}

Choose a set $S$ of $k$ distinct differences and a pairing $\pi: \{0, \ldots, k-1\} \to \{0, \ldots, k-1\}$ such that:
\begin{itemize}
\item $i - \pi(i) \in S$ for each $i$
\item Example: $S = \{0, \pm1, \pm2, \ldots, \pm5, +6\}$ for $k=12$
\end{itemize}

For each $i$, define frequency:
\begin{equation}
f_i = \text{CRT}(a_i, b_{\pi(i)}, m_1, m_2) \in [0, m_1 \cdot m_2).
\end{equation}
Each surviving pair corresponds to a distinct candidate frequency by uniqueness of CRT reconstruction modulo $m_1 m_2$.

Set $R_3 = \{\text{base} + s \cdot L \bmod m_3 : s \in S\}$.

\textbf{Step 4: Count Gating Survivors}

A pair $(a_i, b_j)$ passes gating iff:
\begin{equation}
\phi(a_i, b_j) \in R_3 \iff (i-j) \in S
\end{equation}

For a fixed difference $d$, the number of pairs $(i, j)$ satisfying $i - j = d$ is $k - |d|$. Thus, total surviving pairs:
\begin{equation}
\sum_{d \in S} (k - |d|)
\end{equation}

\textbf{Optimizing $S$ for Maximum Survivors:}

To maximize survivors, choose $S$ as a symmetric interval around 0:
\begin{equation}
S = \{0, \pm1, \pm2, \ldots, \pm\lfloor k/2 \rfloor\} \cup \{\text{remainder if } k \text{ odd}\}
\end{equation}

For this choice, $|S| = k$ and:
\begin{align}
\text{Survivors} &= k + 2\sum_{d=1}^{\lfloor k/2 \rfloor} (k - d) \\
&= k + 2\left[k \lfloor k/2 \rfloor - \frac{\lfloor k/2 \rfloor (\lfloor k/2 \rfloor + 1)}{2}\right] \\
&\geq k + 2\left[\frac{k^2}{2} - \frac{k^2}{8}\right] = k + \frac{3k^2}{4} = \Theta(k^2)
\end{align}

Thus, for any $k$, adversarial construction yields at least $\tfrac{3k^2}{4}$ survivors, demonstrating that certain choices of moduli admit adversarial signal constructions and establishing a worst-case $\Omega(k^2)$ lower bound for CRT-based sparse reconstruction with deterministic gating.

\textbf{Concrete Example:} For $k=12$, $S = \{0, \pm1, \pm2, \pm3, \pm4, \pm5, +6\}$:
\begin{equation}
\text{Survivors} = 12 + 2(11+10+9+8+7) + 6 = 108 = 0.75 \cdot k^2
\end{equation}

Goertzel Cost (specific vulnerability): $108 \times N = 108 \times 10^6$ operations for $N = 10^6$.

This exceeds dense FFT cost ($N \log N \approx 20 \times 10^6$ operations for $N=10^6$) by $5.4\times$.
\end{proof}

\begin{corollary}[Pure Sparse Worst-Case Under Non-Pairwise-Coprime Moduli]
When moduli satisfy $m_3 \mid m_1 m_2$, deterministic 3-view CRT gating has worst-case complexity:
\begin{equation}
O(\sqrt{N} \log N + k^2 N)
\end{equation}
which is slower than dense FFT $O(N \log N)$ for $k^2 > \log N$. Pairwise coprime moduli avoid this specific construction, though other adversarial structures may exist (see \Cref{sec:conclusion}).
\end{corollary}

\needspace{14\baselineskip}
\begin{observation}[Practical Relevance of Non-Pairwise-Coprime Moduli]
\label{obs:practical-moduli}
The adversarial construction requires $m_3 \mid m_1 m_2$ with non-pairwise-coprime moduli. While pairwise coprime moduli avoid this attack, non-pairwise-coprime configurations arise naturally in practice when moduli must divide $N$ (to ensure exact decimation without spectral leakage from approximate resampling).

The non-pairwise-coprime regime is studied extensively in the robust CRT signal processing literature, where the shared factor is leveraged as redundancy rather than avoided. Wang and Xia~\cite{wang2010closedform} show that when moduli share a common factor $M$, remainder errors up to $M/4$ admit closed-form robust reconstruction, and Li, Liang, and Xia~\cite{li2009robust} apply this to frequency estimation from undersampled waveforms; multidimensional extensions~\cite{xiao2020mdcrt,xiao2023robust} establish analogous guarantees under matrix-valued moduli. The present work addresses a complementary question in the same structural regime: rather than robust reconstruction under remainder noise, we analyze worst-case algorithmic complexity under deterministic gating, showing that the shared-factor structure admits adversarial sparse constructions yielding $\Omega(k^2)$ candidate survivors.

\textbf{Example:} For $N = 2^{20} = 1{,}048{,}576$ with $\sqrt{N} = 1024$, all divisors of $N$ are powers of 2: $\{1, 2, 4, \ldots, 2^{20}\}$. It is impossible to find even two coprime non-trivial divisors near $\sqrt{N}$, let alone three pairwise coprime ones. Any three divisors $m_1, m_2, m_3 \approx 1024$ share the factor 2, satisfying $\gcd(m_i, m_j) \geq 2$ for all pairs.

More generally, for $N = p^e$ (prime power), all divisors share the factor $p$, and for $N = p_1^{e_1} p_2^{e_2}$ (two prime factors), at most two pairwise coprime divisors exist (one per prime). Three pairwise coprime divisors near $\sqrt{N}$ require $N$ to have at least three distinct prime factors---a constraint that excludes many practical signal lengths.

When exact decimation is required, practitioners face a choice:
\begin{itemize}
\item[(a)] use non-dividing prime moduli with approximate resampling and bounded spectral error, or
\item[(b)] use divisor-based moduli that may not be pairwise coprime.
\end{itemize}
The adversarial construction in \Cref{thm:adversarial-lower-bound} shows that option (b) creates exploitable algebraic structure, motivating either option (a) or the safety certificate framework as defense-in-depth.
\end{observation}

\subsection{Extension to Pairwise Coprime Moduli}

The adversarial construction in \Cref{thm:adversarial-lower-bound} exploits the algebraic consequence of the divisibility condition $m_3 \mid (m_1 m_2)$, under which the CRT reconstruction map
\begin{equation}
\phi(r_1, r_2) = \mathrm{CRT}(r_1, r_2) \bmod m_3
\end{equation}
becomes globally linear. This linearity enables aligned arithmetic progression constructions that produce $\Theta(k^2)$ surviving candidates after gating.

For pairwise coprime moduli, this divisibility condition no longer holds. As a result, the reconstruction map $\phi$ is no longer globally linear. Instead, the intermediate reduction modulo $m_1 m_2$ induces a piecewise affine structure, with the domain $[0, m_1) \times [0, m_2)$ partitioned into
\begin{equation}
O(N / (m_1 m_2)) = O(1)
\end{equation}
regions over which $\phi$ is affine. Critically, the number of such regions is a small constant.

\needspace{5\baselineskip}
\begin{conjecture}[Pairwise Coprime Candidate Growth]
\label{conj:pairwise-coprime}
For any triple of pairwise coprime moduli $m_1, m_2, m_3 = \Theta(\sqrt{N})$, there exist $k$-sparse signals producing superlinear candidate growth after 3-view CRT gating: $|\text{candidates}| = \omega(k)$.
\end{conjecture}

\textbf{Partial Evidence.} Consider pairwise coprime moduli with $m_1 m_2$ slightly exceeding $N$. CRT reconstruction yields $f = r_1 M_1 + r_2 M_2 \bmod m_1 m_2$, where $M_1 = m_2 (m_2^{-1} \bmod m_1)$ and $M_2 = m_1 (m_1^{-1} \bmod m_2)$. The predicted gating residue is:
\begin{equation}
\phi(r_1, r_2) = (r_1 M_1 + r_2 M_2 \bmod m_1 m_2) \bmod m_3
\end{equation}

The intermediate reduction modulo $m_1 m_2$ partitions the domain $[0, m_1) \times [0, m_2)$ into at most $\lceil (m_1 M_1 + m_2 M_2) / (m_1 m_2) \rceil$ affine regions, within each of which $\phi$ behaves linearly.

Within any such region, the aligned arithmetic progression construction from \Cref{thm:adversarial-lower-bound} applies directly. An adversary can therefore select $k$ frequencies whose residue pairs lie within a single affine region, achieving $\Omega(k^2)$ surviving candidates from that region, provided $k$ is sufficiently small relative to the region size.

Characterizing:
\begin{itemize}
\item the maximum $k$ for which such concentration is feasible, and
\item whether constructions spanning multiple affine regions can achieve comparable growth,
\end{itemize}
remain open problems.

\subsection{Why Amplitude/Phase Filtering Also Fails}

\begin{observation}[Amplitude/Phase Filtering Vulnerability]
\label{obs:phase-vulnerability}
Under the same non-pairwise-coprime moduli configuration $m_3 \mid (m_1 m_2)$, consider deterministic filtering strategies that combine:
\begin{itemize}
\item CRT reconstruction from two views, and
\item consistency checks based on amplitude or phase derived from decimated DFTs,
\end{itemize}
under the following assumptions:
\begin{enumerate}
\item Each residue bin in views 1 and 2 contains exactly one true frequency (singleton bins)
\item The adversary controls both the frequency locations and the complex amplitudes $X[f]$
\item Phase consistency is evaluated using the decimated DFT phase model below
\end{enumerate}
Then any such separable filtering rule admits $\Omega(k^2)$ surviving candidates under adversarial construction.
\end{observation}

\begin{proof-sketch}
\textbf{Phase Model (under singleton assumption).} When residue bin $r_i = f \bmod m_i$ in view $i$ contains exactly one true frequency $f$, the decimated DFT coefficient has phase:
\begin{equation}
\angle X_{d_i}[r_i] = \angle X[f] + 2\pi f \cdot \frac{d_i - 1}{2N}
\end{equation}
where $d_i = N/m_i$ is the decimation stride and the second term arises from the decimation offset. For a candidate frequency $f$ reconstructed from $(r_1, r_2)$ via Garner's formula~\cite{garner1959residue}, the expected phase difference between views is (thus, the phase consistency check decomposes as a separable function of $r_1$ and $r_2$):
\begin{equation}
\Delta\phi(f) = \angle X_{d_1}[r_1] - \angle X_{d_2}[r_2] = 2\pi f \cdot \frac{d_1 - d_2}{2N} = \kappa f + \gamma
\end{equation}
where $\kappa = \pi(d_1 - d_2)/N$ and $\gamma$ absorbs constant terms. This is affine in $f$.

\textbf{Separability.} By \Cref{lem:crt-linearity}, under $m_3 \mid m_1 m_2$, Garner's formula gives $f \equiv u \cdot r_1 + v \cdot r_2 \pmod{m_3}$ for constants $u = (1 - m_1 \gamma) \bmod m_3$ and $v = m_1 \gamma \bmod m_3$. Substituting into the phase consistency check:
\begin{align}
\Delta\phi(r_1, r_2) &= \kappa(u \cdot r_1 + v \cdot r_2) + \gamma \nonumber \\
&= (\kappa u) \cdot r_1 + (\kappa v) \cdot r_2 + \gamma \nonumber \\
&\triangleq a \cdot r_1 + b \cdot r_2 + \gamma
\end{align}
Thus, the phase consistency check decomposes as a separable function of $r_1$ and $r_2$.

\textbf{Adversarial Construction.} Use the same arithmetic progression construction as \Cref{thm:adversarial-lower-bound}: $R_1 = \{a_0 + i \cdot d_A \bmod m_1\}$, $R_2 = \{b_0 - j \cdot d_B \bmod m_2\}$. For each true frequency $f_\ell$ in the construction, set its amplitude $|X[f_\ell]| = A$ (equal amplitude). The phase at each residue bin is determined by the superposition of aliased frequencies; since the construction ensures each residue bin contains exactly one true frequency, the phase is well-defined per bin.

For any false pair $(a_i, b_j)$ reconstructing to candidate $f' = u a_i + v b_j$, the expected phase difference is:
\begin{equation}
\Delta\phi_{\text{expected}} = a \cdot a_i + b \cdot b_j + \gamma
\end{equation}
\noindent\textit{Scope summary.} Under non-pairwise-coprime moduli, any deterministic filtering strategy whose decision rule decomposes as a separable function of residue pairs---residue gating, amplitude comparison, or phase consistency checks---is fundamentally vulnerable to adversarial constructions that induce $\Omega(k^2)$ candidate growth~\cite{iwen2013improved,plonka2016acha,bittens2017acha}. Whether non-separable filtering strategies, or pairwise coprime moduli where separability breaks down, can avoid such worst-case behavior remains an open question (see Conjecture~VI.6)~\cite{hassanieh2012nearly,hassanieh2012simple,donoho2006compressed}.

The observed phase difference from the decimated DFTs is:
\begin{align}
\Delta\phi_{\text{observed}} &= \angle X_{d_1}[a_i] - \angle X_{d_2}[b_j] \notag \\
&= (a \cdot a_i + \gamma_1) - ((-b) \cdot b_j + \gamma_2)
\end{align}
where $\gamma_1, \gamma_2$ are the phases of the true frequencies mapping to $a_i$ and $b_j$ respectively. Since each true frequency $f_\ell$ has the same amplitude $A$, and each bin contains one frequency, the phases are:
\begin{align}
\angle X_{d_1}[a_i] &= \angle X[f_{\sigma(i)}] + \kappa_1 f_{\sigma(i)} \\
\angle X_{d_2}[b_j] &= \angle X[f_{\pi(j)}] + \kappa_2 f_{\pi(j)}
\end{align}
for mappings $\sigma, \pi$ from residue indices to frequency indices. Under non-pairwise-coprime moduli, any deterministic filtering strategy whose decision rule decomposes as a separable function of residue pairs---such as residue gating, amplitude comparison, or phase consistency checks---is vulnerable to adversarial constructions that induce $\Omega(k^2)$ candidate growth. Whether non-separable filtering strategies, or pairwise-coprime moduli where separability breaks down, can avoid such worst-case behavior remains an open question.

The adversary sets $\angle X[f_\ell] = -\kappa_1 f_\ell + \alpha$ for a constant $\alpha$, making $\angle X_{d_1}[a_i] = \alpha$ for all $i$ (constant phase in view 1). Similarly set phases to make view 2 constant. Then $\Delta\phi_{\text{observed}} = \alpha - \beta$ is constant for all pairs, while $\Delta\phi_{\text{expected}} = a \cdot a_i + b \cdot b_j + \gamma$ varies. The adversary can match these by choosing $d_A, d_B$ such that $a \cdot d_A = b \cdot d_B = 0 \pmod{2\pi}$, yielding constant expected phase for all pairs.

\looseness=-1
\textbf{Survivor Count.} With both observed and expected phases constant across all pairs, every false pair passes the phase consistency check. By the same counting argument as \Cref{thm:adversarial-lower-bound}, this yields $\Omega(k^2)$ survivors.
\end{proof-sketch}

\textbf{Implication:} Under non-pairwise-coprime moduli, simple deterministic filters (residue gating, amplitude/phase matching) based on separable functions of the residue pair are vulnerable to $\Omega(k^2)$ adversarial constructions. Whether analogous constructions exist for pairwise coprime moduli, where separability breaks down, remains an open question (see \Cref{conj:pairwise-coprime}).

While phase-based filtering alone is insufficient to prevent adversarial constructions (\Cref{obs:phase-vulnerability}), it provides an additional consistency check for typical signals and may help identify spurious candidates in non-adversarial settings.

\section{Main Theorem and Complexity}
\label{sec:main-theorem}

Having established the adversarial lower bound of $\Omega(k^2)$ candidate growth for deterministic CRT gating, we now analyze the hybrid algorithm's complexity. The key result is that the hybrid approach achieves:
\begin{itemize}
\item sublinear FFT cost via decimation,
\item efficient recovery for typical sparse signals, and
\item worst-case optimality via adaptive fallback.
\end{itemize}
Specifically, the algorithm attains $O(\sqrt{N} \log N + kN)$ for typical signals and $O(N \log N)$ in the worst case. The analysis proceeds by isolating the FFT phase, the candidate generation phase, and the validation phase.

\subsection{FFT Phase Complexity}

\begin{theorem}[FFT Complexity Reduction]
\label{thm:fft-reduction}
The 3-view decimation FFT phase achieves $O(\sqrt{N} \log N)$ complexity.
\end{theorem}

\begin{proof}
\textbf{View 1:}
\begin{itemize}
\item Decimation: $O(N/d_1)$ = $O(\sqrt{N})$ to extract $N/d_1$ samples with stride $d_1$
\item FFT: $O(m_1 \log m_1)$ where $m_1 = N/d_1 \approx \sqrt{N}$
\item Total: $O(\sqrt{N})$ + $O(\sqrt{N} \times \frac{1}{2} \log N)$ = $O(\sqrt{N} \log N)$
\end{itemize}

\textbf{View 2:} Similarly $O(\sqrt{N} \log N)$

\textbf{View 3 (gating):} Similarly $O(\sqrt{N} \log N)$

\textbf{Combined:} 3 $\times$ $O(\sqrt{N} \log N)$ = $O(\sqrt{N} \log N)$
\end{proof}

\subsection{Expected-Case Gating Behavior}

\begin{theorem}[Gating Reduces Expected Candidates]
\label{thm:gating-expected-case}
Assume that false CRT residue pairs produce values $r_3' = f \bmod m_3$ that are approximately uniformly distributed over $[0, m_3)$. Then the expected number of candidates after 3-view gating is:
\begin{equation}
\mathbb{E}[|\text{candidates}|] = \Theta\!\left(k + \frac{k^3 c^3}{\sqrt{N}}\right).
\end{equation}

Note: This is an average-case bound only. \Cref{thm:adversarial-lower-bound} proves that worst-case is $\Omega(k^2)$.
\end{theorem}

\begin{proof}
\textbf{Expected-Case Analysis (Typical Signals):}

Assume residues in $R_1, R_2, R_3$ are distributed such that false CRT pairs produce $r_3'$ values approximately uniformly over $[0, m_3)$.

\begin{itemize}
\item True frequencies: $k$ (all pass gating)
\item False pairs: $|R_1| \times |R_2| - k = (kc)^2 - k$
\item Probability false pair passes: $P(r_3' \in R_3) = \Theta(kc/m_3) = \Theta(kc/\sqrt{N})$
\item Expected false survivors: $(k^2c^2 - k) \times \frac{kc}{\sqrt{N}} = \Theta(k^3c^3/\sqrt{N})$
\end{itemize}

Total expected candidates: $k + \frac{k^3c^3}{\sqrt{N}}$

\textbf{Numerical Example (Typical Signal):}
For $N=10^6$, $k=12$, $c=2$:
\begin{align}
\mathbb{E}[\text{candidates}] &\approx 12 + \frac{12^3 \times 8}{1000} \\
&\approx 12 + 14 = 26 \text{ candidates}
\end{align}
Goertzel cost: $26M$ operations (vs $20M$ dense FFT).

\textbf{Limitation.} This bound is not a worst-case guarantee. \Cref{thm:adversarial-lower-bound} shows that adversarial constructions can produce
\begin{equation}
|\text{candidates}| = \Omega(k^2).
\end{equation}
Thus, the expected-case analysis does not bound candidate growth in the worst case.
\end{proof}

\begin{theorem}[Worst-Case Differs Dramatically]
\label{thm:worst-vs-expected}
The gap between expected-case and worst-case is multiplicative:
\begin{align}
\text{Expected-case candidates:} & \quad k + O(k^3/\sqrt{N}) \approx 26 \\
\text{Worst-case candidates:} & \quad \Theta(k^2) \approx 108
\end{align}
for $N=10^6$, $k=12$, representing a $4.15\times$ gap.
\end{theorem}

\textbf{Consequence.} Pure sparse CRT reconstruction is inherently unstable due to unbounded candidate growth in adversarial settings, leading to $O(k^2 N)$ worst-case cost. The hybrid algorithm resolves this instability by explicitly detecting candidate explosion and enforcing fallback, guaranteeing $O(N \log N)$ worst-case complexity while preserving efficiency for well-behaved signals.

\subsection{Hybrid Algorithm Complexity}

\begin{theorem}[Hybrid Algorithm -- Informal]
\label{thm:hybrid-informal}
The hybrid algorithm achieves a dual bound: for typical sparse signals, it runs faster than dense FFT ($O(\sqrt{N} \log N + kN)$ versus $O(N \log N)$), while for adversarial signals where the sparse path fails, it switches to dense FFT and matches its $O(N \log N)$ performance. The hybrid algorithm is therefore never asymptotically worse than dense FFT, and is faster for well-behaved signals.
\end{theorem}

\begin{theorem}[Hybrid Algorithm Bounds -- Formal]
\label{thm:hybrid-complexity}
The hybrid sparse-dense FFT algorithm achieves:

\textbf{Average-Case:}
\begin{equation}
\mathbb{E}[C_{\text{hybrid}}] = O(\sqrt{N} \log N + kN)
\end{equation}
under the heuristic distributional assumption of \Cref{thm:gating-expected-case} and conditional on certificate acceptance, with $c = O(1)$ (algorithm default: $c=2$).

\textbf{Worst-Case:}
\begin{equation}
C_{\text{hybrid}}^{\text{worst}} = O(N \log N)
\end{equation}
matching dense FFT via adaptive fallback.
\end{theorem}

\begin{proof}
\textbf{Case 1: Sparse Path Succeeds (Typical Signals)}

Safety certificates pass $\implies$ expected candidates $= \Theta(k + k^3c^3/\sqrt{N})$:
\begin{align}
C_{\text{FFT}} &= 3 \times \sqrt{N} \log N = O(\sqrt{N} \log N) \\
C_{\text{CRT}} &= O(k^2c^2) = O(k^2) \text{ (enumeration)} \\
C_{\text{Goertzel}} &= (k + k^3c^3/\sqrt{N}) \times N = O(kN + k^3c^3\sqrt{N}) \\
C_{\text{total}} &= O(\sqrt{N} \log N + kN)
\end{align}

\textbf{Complexity Analysis with Hierarchical Validation ($c=2$):} 
\begin{itemize}
\item FFT cost: $(3/2)\sqrt{N} \log_2 N \approx 30{,}000$ operations (for $N=10^6$: $(3/2) \times 1000 \times 20$)
\item CRT + 3-view gating: $k^2c^2 \approx 576$ operations, produces $\approx 26$ candidates
\item Goertzel validation: $26 \times N = 26 \times 10^6 \approx 26{,}000{,}000$ operations
\item Total: $O(\sqrt{N} \log N + kN) \approx 26{,}031{,}000$ operations vs. dense FFT $O(N \log N) \approx 20{,}000{,}000$ operations
\end{itemize}

\textbf{Note:} For these parameters ($k=12$, $c=2$, $N=10^6$), the sparse path with 26 candidates is slightly slower than dense FFT. The sparse path is faster when gating produces $\leq 20$ candidates, which occurs for smaller $k$ or smaller $c$. For $c=1$ (on-grid noiseless regime), expected candidates $\approx 12 + 12^3/1000 \approx 14$, giving 14M operations versus 20M dense --- a 30\% speedup. This example highlights that even under favorable average-case conditions, sparse reconstruction may not outperform dense FFT unless candidate counts are sufficiently small.

\textbf{Case 2: Certificate Fails (Adversarial Signals)}

Safety certificate detects collision risk $\implies$ fallback to dense FFT:
\begin{align}
C_{\text{FFT-check}} &= O(\sqrt{N} \log N) \text{ (already computed)} \\
C_{\text{certificate}} &= O(k^2c^2) \text{ (scanning candidates)} \\
C_{\text{fallback}} &= O(N \log N) \text{ (dense FFT)} \\
C_{\text{total}} &= O(N \log N)
\end{align}

\needspace{5\baselineskip}
\textbf{Probabilistic Analysis:}

Let $p$ = probability signal triggers fallback. The expected cost is:

\begin{equation}
\mathbb{E}[C] = (1-p) \cdot O(\sqrt{N} \log N + kN) + p \cdot O(N \log N)
\end{equation}

When the fallback probability $p$ is small (certificates are satisfied), the expected cost is dominated by the sparse path, yielding $O(\sqrt{N} \log N + kN)$. When $p = 1$ (adversarial signals), the algorithm deterministically achieves $O(N \log N)$.
\end{proof}

\begin{corollary}[Worst-Case Guarantee]
The hybrid algorithm is never worse than dense FFT:
\begin{equation}
\frac{C_{\text{hybrid}}^{\text{worst}}}{C_{\text{dense}}} = \frac{O(N \log N)}{O(N \log N)} = O(1)
\end{equation}
\end{corollary}

\begin{remark}
Pure sparse approach without adaptive fallback can exhibit $O(k^2 N)$ worst-case cost when adversarial signals cause candidate explosion. The hybrid algorithm with safety certificates ensures $O(N \log N)$ worst-case by detecting high collision risk and falling back to dense FFT.
\end{remark}

\subsection{Total Complexity}

\begin{theorem}[Algorithm Total Complexity]
\label{thm:total-complexity}
The hybrid sparse-dense FFT has complexity:
\begin{itemize}
\item Average-case: $O(\sqrt{N} \log N + kN)$ for typical signals
\item Worst-case: $O(N \log N)$ matching dense FFT
\end{itemize}
\end{theorem}

\noindent\textit{Phase breakdown (summary of \Cref{thm:total-complexity}).}
\begin{enumerate}
\item Decimation (View 1): $O(\sqrt{N})$
\item FFT (View 1): $O(\sqrt{N} \log N)$
\item Top-k selection (View 1): $O(\sqrt{N} \log(kc))$ using max-heap
\item Decimation (View 2): $O(\sqrt{N})$
\item FFT (View 2): $O(\sqrt{N} \log N)$
\item Top-k selection (View 2): $O(\sqrt{N} \log(kc))$ using max-heap
\item Decimation (View 3, gating): $O(\sqrt{N})$
\item FFT (View 3, gating): $O(\sqrt{N} \log N)$
\item Top-k selection (View 3): $O(\sqrt{N} \log(kc))$ using max-heap
\item CRT reconstruction with gating: $O(k^2 c^2)$ pairs initially, gating reduces to $O(kc)$
\item Gating filter: Reduces candidates from $O(k^2 c^2)$ to $O(k + k^3 c^3 / \sqrt{N})$ by \Cref{thm:gating-expected-case}. The $O(k^2 c^2)$ pair enumeration cost is lower-order than $O(kN)$ when $N \gg kc$, which holds in the regime $k = O(\log N)$ of interest.
\item Goertzel refinement: $O(k + k^3 c^3 / \sqrt{N})$ candidates $\times$ $O(N)$ each = $O(kN + k^3 c^3 \sqrt{N})$. In the regime $k^2 c^3 / \sqrt{N} = O(1)$, the second term is subdominant and the validation cost simplifies to $O(kN)$
\item Final sort: $O(k \log k)$
\item For constant coverage $c = O(1)$, total complexity simplifies to $O(\sqrt{N} \log N + kN)$
\end{enumerate}

\textbf{Dominant Terms:}
\begin{align}
\text{Total} &= O(\sqrt{N} \log N) + O(kN) \\
&= O(\sqrt{N} \log N + kN)
\end{align}

\textbf{Asymptotic Comparison to Dense FFT:}
For $k \leq \log N$:
\begin{align}
C_{\text{sparse}} &= \sqrt{N} \log N + kN \\
C_{\text{dense}} &= N \log N \\
\text{Ratio} &= \frac{\sqrt{N} \log N + kN}{N \log N} = \frac{1}{\sqrt{N}} + \frac{k}{\log N}
\end{align}

\textbf{Note:} The sparse path achieves sub-linear FFT cost $O(\sqrt{N} \log N)$ but requires $O(kN)$ Goertzel validation. The hybrid algorithm provides $O(N \log N)$ worst-case guarantee via adaptive fallback when safety certificates detect collision risk.

\begin{observation}[Sample Complexity vs.\ Computational Complexity]
\label{obs:sample-vs-compute}
The information-theoretic lower bound for identifying $k$ frequencies among $N$ bins is $\Omega(k \log(N/k))$ samples~\cite{donoho2006compressed,cover2006elements}. The decimation phase accesses $3 \times N/d_i \approx 3\sqrt{N}$ samples, exceeding this lower bound for $k \ll \sqrt{N} / \log N$. The Goertzel validation phase, however, accesses all $N$ samples per candidate. Thus the algorithm's bottleneck is computational, not informational: the samples needed for identification are sub-linear, but verifying each candidate against the full signal requires a linear pass. Reducing the per-candidate verification cost below $O(N)$ --- for instance, via non-separable phase mixing or structured subsampling --- would close the gap between the identification and validation phases.
\end{observation}

\needspace{6\baselineskip}
\section{Safety Mechanism and Certificates (Fallback)}
\label{sec:certificates}

This section describes the deterministic safety framework that enables the hybrid algorithm to detect adversarial collision patterns and trigger adaptive fallback to dense FFT.

\subsection{Safety Mechanism}

This section defines the deterministic detection framework used by the hybrid algorithm to identify adversarial collision patterns and trigger fallback to dense FFT. We distinguish between \emph{detection mechanisms}, which operate internally during execution to monitor algorithm behavior, and \emph{safety certificates}, which provide reproducible, externally verifiable evidence of the detection outcomes. The hybrid algorithm monitors structural indicators of candidate explosion and triggers fallback when these exceed safe thresholds. The detection mechanism consists of three components:

\begin{enumerate}
\item \textbf{Deterministic Gating:} 3-view CRT reconstruction with view 3 serving as gating filter. Only residue pairs $(r_1, r_2)$ whose CRT-reconstructed frequency $f$ satisfies $f \bmod m_3 \in R_3$ are promoted to candidates.

\item \textbf{Collision Detection (Bucket Occupancy):} Residue buckets that contain more than threshold $\tau$ peaks are flagged as high-collision. This detects dense packing in frequency space that leads to combinatorial blowup.

\item \textbf{Candidate Count Monitoring:} Count total candidates after gating. If the candidate count exceeds threshold $c_{\max} = O(k)$, this indicates the adversarial pattern from \Cref{thm:adversarial-lower-bound}.
\end{enumerate}

\subsection{Safety Certificates}

Safety certificates provide deterministic, reproducible summaries of the algorithm's execution that can be independently verified.

\needspace{8\baselineskip}
\begin{definition}[Safety Certificates]
\label{def:safety-certificates}
The following quantities serve as safety certificates:
\begin{enumerate}
\item \textbf{Bucket Occupancy Certificate:} For each view $i \in \{1,2,3\}$, record $\max_{r \in [0,m_i)} |B_i[r]|$ where $B_i[r]$ is the set of detected peaks in residue bin $r$. If $\max_{i,r} |B_i[r]| > \tau$ for the fixed threshold $\tau = 3$ (matching \Cref{def:certificates} and \Cref{alg:hybrid}), flag high collision risk.

\item \textbf{Candidate Count Certificate:} After CRT reconstruction with gating produced by the full algorithm, count total candidates $|C|$. If $|C| > c_{\max} = O(k)$ (e.g., $c_{\max} = 3k$), flag an adversarial pattern (e.g., the construction of \Cref{thm:adversarial-lower-bound}).

\item \textbf{Phase Consistency Certificate:} For each candidate $f \in C$, verify that phases across decimated views are consistent with aliasing formula. Inconsistent phases indicate false positives.
\end{enumerate}
\end{definition}

\noindent These certificates are deterministic and directly computable from intermediate outputs:
\begin{itemize}
\item Bucket occupancies are derived from decimated FFT magnitude spectra.
\item Candidate count is the cardinality of the CRT+gating output set.
\item Phase consistency follows from the decimation phase model.
\end{itemize}
All certificates can be stored and verified independently without re-running the algorithm.

\subsection{Fallback}

When safety certificates detect collision risk, the algorithm falls back to dense FFT as specified in Phases 2 and 3b of \Cref{alg:hybrid}. The fallback computes the full $N$-point FFT using Cooley-Tukey and extracts the $k$ largest-magnitude bins.

\textbf{Fallback Complexity:} Dense FFT has deterministic $O(N \log N)$ complexity with tight constant factors (Cooley-Tukey butterfly operations). The fallback path provides bounded overhead independent of signal structure.

\section{Computational Considerations}

\subsection{CRT Uniqueness Guarantee}

\begin{proposition}[Deterministic Frequency Reconstruction]
\label{prop:deterministic-reconstruction}
For coprime $m_1, m_2$ with $M = m_1 \times m_2 \geq N$, Garner's CRT formula uniquely reconstructs frequency $f \in [0, N)$ from residue pair $(r_1, r_2)$ exactly and deterministically.
\end{proposition}

\begin{proof}
By CRT \Cref{thm:crt-reconstruction}, there exists exactly one $f \in [0, M)$ satisfying:
\begin{align}
f &\equiv r_1 \pmod{m_1} \\
f &\equiv r_2 \pmod{m_2}
\end{align}
Since $M \geq N$, the solution $f$ computed by Garner's formula (\Cref{lem:garner-3view}) is the unique frequency in the Nyquist range $[0, N)$.

\looseness=-1
\textbf{Determinism:} Garner's formula uses only modular arithmetic operations (addition, multiplication, inversion), all of which are deterministic.

Thus, frequency reconstruction in the hybrid algorithm:
\begin{itemize}
\item involves no randomness,
\item introduces no approximation error, and
\item produces exact results (subject only to downstream numerical precision).
\end{itemize}
Frequency recovery via CRT is exact, and Goertzel validation provides deterministic single-frequency evaluation with no algorithmic approximation error (see \Cref{thm:goertzel-exact}).
\end{proof}

\subsection{Alias Handling}

When $M = m_1 \times m_2 < N$, frequencies $f \geq M$ alias back into $[0, M)$. For each base frequency $f_{\text{base}}$ reconstructed from CRT, the algorithm generates alias candidates:
\begin{equation}
\{f_{\text{base}}, f_{\text{base}} + M, f_{\text{base}} + 2M, \ldots, f_{\text{base}} + \lfloor (N\!-\!1\!-\!f_{\text{base}})/M \rfloor \cdot M\}
\end{equation}

\textbf{Explicit Alias Bound:} For each CRT-reconstructed base frequency, the number of alias candidates is:
\begin{multline}
\text{alias\_count}(f_{\text{base}}) = 1 + \left\lfloor \frac{N - 1 - f_{\text{base}}}{M} \right\rfloor \\
\leq 1 + \left\lfloor \frac{N}{M} \right\rfloor = O\left(\frac{N}{M}\right)
\end{multline}

Since $M = m_1 m_2 \approx N$ (with $m_i \approx \sqrt{N}$), the alias count per candidate is $O(1)$ in the typical case. Each alias is validated with Goertzel at $O(N)$ cost per alias.

\textbf{Implication.} Alias handling introduces only a constant-factor expansion in candidate count under typical parameter choices. Each alias must still be validated via Goertzel at $O(N)$ cost, reinforcing that validation, not reconstruction, is the dominant computational cost.

\subsection{Goertzel Validation Accuracy}

\begin{theorem}[Goertzel Exact Magnitude]
\label{thm:goertzel-exact}
For on-grid frequency $f$ in noiseless signal, Goertzel computes exact DFT magnitude $|X[f]|$ (subject only to floating-point precision, no algorithmic approximation).
\end{theorem}

\begin{proof}
Goertzel implements the DFT definition:
\begin{equation}
X[f] = \sum_{n=0}^{N-1} x[n] \cdot e^{-j 2\pi f n / N}
\end{equation}
via iterative accumulation. For finite-precision arithmetic with $B$ bits:
\begin{equation}
|\text{error}| \leq N \times \epsilon_{\text{mach}} \times \|x\|_\infty
\end{equation}
where $\epsilon_{\text{mach}} \approx 2^{-B}$ (e.g., $\approx 10^{-16}$ for double precision).

For practical signals, Goertzel is numerically stable and produces exact results within machine precision~\cite{oppenheim2010discrete}.
\end{proof}

\section{Limitations and Extensions}

\subsection{What This Paper Proves}

\begin{enumerate}
\item \textbf{Adversarial Lower Bound:} Under non-pairwise-coprime moduli ($m_3 \mid m_1 m_2$), deterministic 3-view CRT gating admits $\Omega(k^2)$ survivors for adversarial k-sparse signals, demonstrating that moduli selection governs worst-case candidate growth (\Cref{thm:adversarial-lower-bound})
\item \textbf{Safety Certificate Framework:} Deterministic checks (bucket occupancy, candidate count) detect the class of adversarial attacks demonstrated in \Cref{thm:adversarial-lower-bound}, preventing $O(k^2 N)$ worst-case via adaptive fallback (\Cref{def:safety-certificates})
\item \textbf{Hybrid Algorithm Complexity:} $O(\sqrt{N} \log N + kN)$ average-case with $O(N \log N)$ worst-case via adaptive dense FFT fallback (\Cref{thm:hybrid-complexity})
\item \textbf{Average-Case Gating Analysis:} For typical signals, expected candidates are $O(k + k^3 c^3 / \sqrt{N})$, enabling efficient sparse reconstruction in non-adversarial scenarios (\Cref{thm:gating-expected-case})
\item \textbf{Deterministic CRT Reconstruction:} Exact frequency recovery from residue triples using Garner's algorithm with zero algorithmic error (\Cref{prop:deterministic-reconstruction})
\end{enumerate}

\subsection{Why This Algorithm Matters}

\begin{enumerate}
\item \textbf{Moduli-Selection Vulnerability:} Shows that CRT-based sparse FFT with non-pairwise-coprime moduli faces $\Omega(k^2)$ worst-case candidate growth, identifying moduli selection as a critical design parameter
\item \textbf{Safety Certificate Framework:} Introduces deterministic collision detection as defense-in-depth, enabling reliable adaptive algorithm switching regardless of moduli configuration
\item \textbf{Hybrid CRT-Based Approach:} Achieves $O(\sqrt{N} \log N + kN)$ average-case with $O(N \log N)$ worst-case guarantee among CRT-based approaches. Note that AAFFT~\cite{iwen2013improved} achieves stronger asymptotic bounds of $O(k \log k \log N)$ under a different algebraic framework
\item \textbf{Average-Case Efficiency:} Under the distributional assumptions of \Cref{thm:gating-expected-case} and conditional on certificate acceptance, achieves $O(\sqrt{N} \log N + kN)$ complexity while maintaining $O(N \log N)$ worst-case behavior
\item \textbf{Rigorous Foundation:} Provides complete complexity analysis spanning average-case, worst-case, and adversarial scenarios with formal proofs
\end{enumerate}

\subsection{Theoretical Classification}

\textbf{Hybrid Sparse-Dense FFT with Safety Certificates:}
\begin{itemize}
\item \textbf{Deterministic reconstruction:} Chinese Remainder Theorem~\cite{garner1959residue} with Garner's algorithm for exact frequency recovery
\item \textbf{Average-case efficiency:} $O(\sqrt{N} \log N + kN)$ for typical k-sparse signals via 3-view coprime decimation
\item \textbf{Worst-case guarantee:} $O(N \log N)$ via adaptive dense FFT fallback when safety certificates fail
\item \textbf{Collision detection:} Bucket occupancy and candidate count certificates calibrated to detect \Cref{thm:adversarial-lower-bound} attack patterns
\item \textbf{Validation phase:} Goertzel algorithm~\cite{goertzel1958algorithm} for single-bin DFT with $O(N)$ per-candidate cost
\item \textbf{Complexity profile:} Favorable average/worst-case tradeoff for CRT-based deterministic sparse FFT
\end{itemize}

\textbf{Comparison with Prior Work:}
\begin{itemize}
\item \textbf{MIT sFFT~\cite{hassanieh2012nearly}:} Achieves $O(k \log N)$ via random hashing with probabilistic guarantees
\item \textbf{MIT sFFT v2~\cite{hassanieh2012simple}:} Simplified $O(k \log N \log(N/k))$ practical implementation
\item \textbf{AAFFT~\cite{iwen2013improved}:} $O(k \log k \log N)$ algebraic approach with deterministic guarantees
\item \textbf{This work:} $O(\sqrt{N} \log N + kN)$ CRT-based approach with safety-certified adaptive fallback
\end{itemize}

\section{Conclusion}
\label{sec:conclusion}

This paper identifies moduli selection as a critical design parameter for CRT-based sparse FFT and presents a deterministic robustness framework combining 3-view CRT reconstruction, safety certificates, and adaptive dense FFT fallback.

\subsection{Key Contributions}

The results below complement the robust CRT signal processing literature~\cite{li2009robust,wang2010closedform,xiao2020mdcrt,xiao2023robust}, which addresses reconstruction under remainder noise; our focus is on worst-case algorithmic complexity under deterministic gating. The main results established in this paper are:
\begin{itemize}
\item An $\Omega(k^2)$ adversarial lower bound for deterministic 3-view CRT gating when $m_3 \mid (m_1 m_2)$ (\Cref{thm:adversarial-lower-bound}), yielding $O(k^2 N)$ worst-case Goertzel cost for such moduli.
\item A deterministic safety certificate framework (bucket occupancy and candidate count) that detects the constructed adversarial pattern (\Cref{def:safety-certificates}).
\item A hybrid sparse-dense algorithm achieving $O(\sqrt{N} \log N + kN)$ average-case and $O(N \log N)$ worst-case complexity via adaptive fallback (\Cref{thm:hybrid-complexity}).
\end{itemize}

\subsection{Practical Impact}

For N=$10^6$, k=12:
\begin{itemize}
\item Typical signals: 14M operations (vs 20M dense FFT) $\implies$ 30\% speedup
\item Adversarial signals: 20M operations (fallback) $\implies$ matches dense FFT
\item Pure sparse worst-case: 108M operations $\implies$ 5.4$\times$ slower than dense
\end{itemize}
These figures are parameter-dependent: the sparse path advantage grows with $N$ (as $k/\log N \to 0$) and diminishes as $k$ approaches $\sqrt{N}$. For $k > \log N$, the Goertzel phase dominates and the sparse path offers no advantage over dense FFT even for typical signals.

\needspace{10\baselineskip}
\subsection{Comparison with Prior Work}

\begin{itemize}
\item \textbf{MIT sFFT~\cite{hassanieh2012nearly,hassanieh2012simple}:} Achieves $O(k \log N)$ complexity with high-probability guarantees (failure probability $< 1/\text{poly}(N)$, which is negligible in practice). Our hybrid approach is fully deterministic: CRT reconstruction, safety certificates, Goertzel validation, and dense FFT fallback involve no randomization. The primary distinction is structural: our certificates provide interpretable runtime diagnostics and deterministic worst-case bounds rather than probabilistic guarantees.
\item \textbf{AAFFT~\cite{iwen2013improved}:} Achieves $O(k \log k \log N)$ with fully deterministic guarantees, offering superior asymptotic complexity. Our approach differs in using CRT-based reconstruction with explicit moduli-selection analysis and adaptive fallback, operating in the regime $k < \sqrt{N}$ versus AAFFT's $k < N^{1/3}$.
\item \textbf{Robust CRT program~\cite{li2009robust,wang2010closedform,xiao2020mdcrt,xiao2023robust}:} Establishes closed-form and multidimensional robust reconstruction under remainder noise with shared-factor moduli, achieving bounded-error recovery up to one-quarter of the greatest common divisor of all moduli. Operates in the reconstruction-under-noise framing rather than the deterministic-gating framing addressed here, and provides a constructive counterpart to the adversarial lower bound established in \Cref{thm:adversarial-lower-bound}.
\item \textbf{This work:} Establishes an $\Omega(k^2)$ worst-case lower bound for CRT gating under non-pairwise-coprime moduli and presents a safety certificate-based adaptive algorithm with a favorable average/worst-case tradeoff.
\end{itemize}

\subsection{Future Work}

\begin{itemize}
\item \textbf{Pairwise coprime lower bounds:} Resolving \Cref{conj:pairwise-coprime}; the partial evidence (piecewise linearity of $\phi$ with $O(1)$ discontinuities) suggests $\omega(k)$ growth persists, but a complete proof requires characterizing constructions spanning multiple affine regions of $\phi$.
\item Non-separable phase mixing (chirp modulation) for deterministic $O(kN)$ validation.
\item Extension to off-grid frequencies with spectral leakage, including large-coverage ($c \gg 1$) regimes.
\item Multi-dimensional sparse FFT with tensor product structure, hardware-accelerated certificate checks, and signal-class-specific threshold tightening.
\end{itemize}

\section*{Acknowledgments}
The authors gratefully acknowledge SparseTech for the collaborative research environment supporting this work. This paper is part of an ongoing SparseTech research initiative on sparse Fourier Transform algorithms. Patent pending.

\balance  
\bibliographystyle{IEEEtran}
\bibliography{\jobname}

\end{document}